\documentclass[article]{IEEEtran}

\usepackage{fancyhdr}
\usepackage{subcaption}
\usepackage{physics}
\usepackage{amsmath,amsthm,amssymb}
\usepackage{filecontents}
\usepackage{siunitx}
\usepackage{graphicx} 
\usepackage{optidef}
\usepackage{bbm}
\usepackage{amsmath}
\usepackage{amsfonts}
\usepackage{optidef}
\usepackage{algorithmic,eqparbox,array}
\usepackage{xcolor}
\usepackage{tikz}
\usepackage{lipsum}

\usetikzlibrary{decorations.pathreplacing,calc}
\usepackage[ruled,vlined]{algorithm2e}
\usepackage{cite}[sorting=none]

\theoremstyle{definition}
\theoremstyle{remark}

  \newtheorem{lemma}{Lemma}

\newtheorem{remark}{Remark}

\begin{document}
\title{Information Bottleneck Meets Quantization: \\Finite Rate Analysis and Optimal Designs}
\author{Francesco Binucci, and Paolo Banelli \thanks{This work was supported by Next Generation EU under the Italian NRRP, Mission 4, Component 2, Investment 1.3, CUP E83C22004640001, partnership on “Telecommunications of the Future” (PE00000001 - program “RESTART”). 
The authors are with the Engineering Department, University of Perugia, Via G. Duranti, 93 06125, Italy, (email: \{francesco.binucci,paolo.banelli\}@unipg.it).
This work is a substantial extension of the preliminary results shown in \cite{eusipco2026accepted}.
}}
\maketitle

\begin{abstract} The Information Bottleneck (IB) is a well established framework that looks for a latent compact representation of a data source, by trading rate and data-size representation, for information accuracy with respect to another target data. The Gaussian IB (GIB) is its simple closed form solution, when the target is jointly Gaussian with the source. Actually, in many practical problems the latent representation has to be stored or represented by a finite number of bits, while the optimal (G)IB solution has not. First, this manuscript theoretically analyzes the effect of scalar and vector quantization of the GIB latent representation, and its impact on the (dis)informativeness with respect to the target data. Then, task-oriented quantization designs are proposed by (jointly) reformulating the GIB optimization problem under a finite-rate constraint on the latent representation. Simulation results on MMSE regression problems confirm the effectiveness of the proposed quantization designs, which show significant gains with respect to more heuristic, or separate, quantization designs of the standard GIB latent representation. Finally, the paper extends the task-oriented philosophy to non-Gaussian settings, by properly modifying the cost function used in variational auto-encoders (VAEs) of IB-inspired vector quantizers.       
\end{abstract}

\section{Introduction}

\subsection{Context and Motivation}
The Information Bottleneck (IB) principle \cite{tishby2000information} is a fundamental information-theoretic framework, rooted in rate--distortion theory \cite{cover1999elements}, that aims at extracting a compressed representation of a random variable while preserving the information that is relevant for a given inference task. Unlike classical source coding approaches, whose focus is  minimizing distortion for source reconstruction, the IB principle seeks to minimize the representation (size) cost  subject to a constraint on the task-relevant information, typically quantified in terms of mutual information with the target variable.

This paradigm has become a cornerstone in the design of task/goal-oriented (GO) communication systems \cite{strinati20216g, hu2024asurvey}, where the objective is to transmit only the information that is strictly necessary to perform a downstream task, rather than faithfully reconstructing the original source. Such an approach enables significant savings in both communication and computational resources, making it particularly attractive for emerging applications in intelligent networks and distributed inference systems \cite{di2023go,gunduz2023beyond}.

The IB principle has been extensively investigated in recent years, demonstrating strong effectiveness in supervised feature extraction and GO communication frameworks \cite{shao2021learning, sun2023adaptive, xie2023robust,shao2022task,hassanpur2025Adeep}. However, most existing works on IB-based GO communications rely on analog transmission schemes, where the extracted features are directly modulated and transmitted over the channel \cite{gunduz2023beyond,shao2022task,sun2023adaptive}. While analytically convenient, these approaches are not readily compatible with practical communication systems, which predominantly rely on digital modulation and coding schemes.

Despite recent efforts on the design of digital task-oriented and semantic communication schemes \cite{huang2025d,liu2024ofdm}, only few of them have explicitly addressed the digital implementation of IB-based strategies, as in \cite{xie2023robust}. Anyway, independently of digital communication, and possibly  storage, a comprehensive theoretical understanding of quantization strategies for IB-based representations, remains largely unexplored. This is probably due to the intrinsic difficulty of the problem, which involves information-theoretic quantities that, even for the non-quantized IB problem, generally do not admit closed-form expressions. A notable exception is the multivariate Gaussian regression setting, where the classical IB problem admits a closed-form solution, referred to as the Gaussian Information Bottleneck (GIB) \cite{chechik2003Gaussian}.

In this work, we firstly build upon the GIB framework to investigate the impact of finite-rate quantization on IB-extracted representations. Specifically, we analyze both scalar and vector quantization schemes, and then we propose more effective task-oriented bit allocation strategies, under limited rate budgets.
Furthermore, inspired by  vector-quantized variational autoencoders (VQ-VAE) \cite{van2017neural}, we extend the proposed framework to non-Gaussian settings by  incorporating task-relevant information measures in the VQ-VAE design. 

\subsection{Related Work}
The quantization problem for  information bottleneck (IB) frameworks has been investigated in \cite{meidlinger2014relation,Rate-information-optimal-GIB}, where the relationship between classical rate--distortion (R--D) theory and the Gaussian Information Bottleneck (GIB) is analyzed for channel-output estimation tasks. Therein it is shown that the GIB solution can be interpreted as an R--D compression scheme, preceded by a square-root Wiener pre-filter. However, such results are restricted to channel-output estimation problems and do not generalize to broader inference settings involving the prediction of target variables from high-dimensional observations.

A related research direction \cite{hassanpour2020forward,ReducedComplexityOptimizationIB,mahvari2021scalable} considers IB-based quantization for distributed estimation under communication-rate constraints. These methods typically rely on iterative optimization procedures and/or high-dimensional discrete mappings, such as vector quantizers, resulting in significant computational complexity and limited analytical tractability.

The quantization problem has also been investigated within the deterministic information bottleneck (DIB) framework~\cite{strouse2017deterministic}, where the IB compression term $I(X;Z)$ is replaced by the entropy term $H(Z)$. This formulation enables deterministic encoders and simplifies optimization in discrete settings. Extensions to continuous and Gaussian variables have been proposed in~\cite{AParamtretricInformationBottleneck}; however, such approaches do not exploit the closed-form structure of the GIB and do not provide analytical solutions for rate-constrained quantization and bit allocation.

In practical machine learning applications, the IB principle is commonly implemented through neural architectures, such as the variational information bottleneck (VIB)~\cite{alemi2016deep}. These methods enable digital implementations of IB-based edge inference systems. For example, \cite{xie2023robust} proposes a VIB-based framework in which compressed latent representations are vector-quantized and transmitted to an edge server through a modulation and coding scheme (MCS). However, the quantization stage is treated as non-trainable, and the problem of optimal vector quantization is not addressed. A similar methodology is adopted in~\cite{zhang2026vector} for over-the-air computation scenarios.

Another relevant line of research is represented by VQ-VAE  approaches~\cite{van2017neural}. These methods rely on encoder--decoder architectures jointly trained with a vector-quantized codebook and have recently attracted interest in task-oriented communication systems. While many existing approaches remain rooted in reconstruction-oriented representation learning~\cite{van2017neural,meng2025channel,lyu2025vq}, some recent works incorporated downstream task objectives for classification-oriented semantic communications ~\cite{Hu2023Robust, chao2025task}. Nevertheless, these methods still rely on generative reconstruction-driven paradigms and do not explicitly address information-theoretic rate allocation and quantization design for IB-based representations.

Finally, the digital encoding of IB-extracted features has been investigated in~\cite{binucci2024opportunistic,binucci2024asilomar}, where the GIB solution is shown to implicitly reduce the entropy of the compressed representation, thereby enabling communication savings in goal-oriented communication systems. However, these works do not explicitly consider rate constraints or optimal quantization schemes.

\subsection{Our Contribution}
Extending the preliminary results derived in \cite{eusipco2026accepted}, we provide a deeper theoretical characterization and alternative designs for scalar and vector quantizers of latent representations produced by the (Gaussian) Information Bottleneck. The main contributions are summarized as follows:
\begin{itemize}
    \item Targeting multivariate Gaussian regression tasks, we thoroughly present the GIB-aided feature selection and quantization framework proposed in \cite{eusipco2026accepted}, leading to closed-form expressions for both the optimal encoder and the optimal rate allocation across components in separate (per-component) quantization schemes.
    
    \item Differently from \cite{eusipco2026accepted}, we show that, in a task-oriented design, it is preferable a rate allocation strategy that maximizes the mutual information between the latent representation and the target, rather than optimizing the standard IB objective. 
    
    \item Then, we extend the proposed scalar (bit) rate allocation framework to block-vector quantizers (VQs), where the features are grouped into blocks and jointly quantized. 
    
    \item We extend the approach to non-Gaussian settings, by proposing a data-driven quantizer that exploits a task-oriented reformulation of VQ-VAEs \cite{van2017neural}.
    
    \item We validate the proposed designs through simulations on both synthetic and real-world datasets, demonstrating that GIB- and MI-aided designs consistently outperform conventional reconstruction-oriented quantization schemes, while closely matching the theoretical predictions.
\end{itemize}
\subsection{Paper Organization}
The remainder of this paper is organized as follows. Section~\ref{sec:preliminaries} reviews the GIB framework and analyzes  a distortion-driven quantization of the GIB latent representation, providing a theoretical characterization of the induced quantization noise. Section~\ref{sec:quantization} first proposes a GIB-matched quantization method, that jointly optimizes the number of quantization bits and the GIB compression matrix, to achieve an effective trade-off between compression efficiency and inference performance. Then, 
the performance results of this strategy naturally leads to the improved MI-aided version, which better focuses on the informativeness of the quantized latent representation with respect to the target variable. Section~\ref{sec:block_vector_vq} extends the designs to vector quantizers, and to non-Gaussian settings by a task-oriented VQ-VAE. Section~\ref{sec:sim_res} reports simulation results that validate the proposed design for MMSE regression tasks, and the performance gain with respect to more heuristic and separate designs of the GIB and the quantizer. 
Finally, Section~\ref{sec:conclusion} concludes the paper and outlines directions for future work. We use lower- and upper-case letters to indicate column vectors and matrices, respectively, $p(\mathbf{x})$ to represent the associated probability density function.
\section{Gaussian Information Bottleneck Quantization}\label{sec:preliminaries}
This section summarizes the IB principle for Gaussian random variables, and analyzes the effect of a finite-precision representation of the latent features extracted by the conventional GIB approach, considering both scalar-quantization (SQ) (i.e., per-component) and vector-quantization (VQ) schemes. This theoretical foundation, will be further exploited in the next section to propose a quantization-aware GIB scheme, e.g., under a finite bit-budget constraint.

\subsection{Gaussian Information Bottleneck}\label{sec:gaussian_ib_section}

Consider a pair of random variables $\mathbf{x} \in \mathbb{R}^{n_x}$, and $\mathbf{y} \in \mathbb{R}^{n_y}$ denoting, respectively, the input and output of an inference task. The Information Bottleneck (IB) principle seeks a (probabilistic) mapping of $\mathbf{x}$ on a compact representation $\mathbf{z}$, which retains as much information as possible about the task variable $\mathbf{y}$ \cite{tishby2000information}. This goal can be formulated as
\begin{equation}\label{eq:ib_prob}
    \min_{p(\mathbf{z}|\mathbf{x})} I(\mathbf{z};\mathbf{x}) - \beta I(\mathbf{z};\mathbf{y}),
\end{equation}
where the mutual information $I(\mathbf{z};\mathbf{x})$ quantifies the degree of source compression, while the mutual information $I(\mathbf{z};\mathbf{y})$ captures the relevance of the compressed representation $\mathbf{z}$ to the (learning) target $\mathbf{y}$. The parameter $\beta$ is a Lagrange multiplier that tunes the compression--relevance trade-off \cite{tishby2000information}.

Assuming that $(\mathbf{x},\mathbf{y})$, are jointly characterized by a zero-mean multivariate Normal distribution,
%
%
with covariance matrices $\mathbf{\Sigma_x}$, $\mathbf{\Sigma_y}$, $\mathbf{\Sigma_{xy}}$, and conditional covariance matrix $\mathbf{\Sigma_{x|y}}=\mathbf{\Sigma_{x}}-\mathbf{\Sigma_{xy}\Sigma_{y}^{\text{-1}}\Sigma_{yx}}$, the optimal IB solution admits a linear stochastic encoding, expressed by \cite{chechik2003Gaussian}
\begin{equation}\label{eq:gib_mapping}
    \mathbf{z}=\mathbf{A}_{\beta}\mathbf{x}+\boldsymbol{\xi},
\end{equation}
which is also jointly Gaussian with $(\mathbf{x},\mathbf{y})$, and where $\boldsymbol{\xi}\sim\mathcal{N}(\mathbf{0},\mathbf{I})$ is a multivariate Gaussian noise term independent of $\mathbf{x}$. The projection matrix  $\mathbf{A}_{\beta} \in \mathbb{R}^{n_z \times n_x}$, is constructed from the eigenpairs $\{\lambda_i,\mathbf{v}_i\}_{i=1}^{n_x}$ of the CCA matrix $\mathbf{\Sigma}_{\mathbf{x}|\mathbf{y}}\mathbf{\Sigma}_{\mathbf{x}}^{-1}$ 
\begin{equation}
\label{eq:GIB-CompressionMatrix}
\mathbf{A}_{\beta}=
\bigl[\,\alpha_{1}\mathbf{v}_{1}^{{T}};\alpha_{2}\mathbf{v}_{2}^{{T}};\dots;\alpha_{n_z}\mathbf{v}_{n_z}^{{T}}\,\bigr], \hspace{6pt}
\beta_{n_z}^{c} < \beta \le \beta_{n_z+1}^{c},
\end{equation}
where $\{\mathbf{v}_i\}$ denotes the left (unitary) eigenvectors, ordered according to theirs  eigenvalues $0<\lambda_i\leq \lambda_{i+1} \leq 1$, $n_z\leq n_x$, $\beta_i^{c}=\frac{1}{1-\lambda_i}$ are the critical values for $\beta$ to activate a CCA direction, $\alpha_{i}=\sqrt{\frac{\beta(1-\lambda_{i})-1}{\lambda_{i}p_{i}}}$ are the loading factor for each component, and $p_{i}=\mathbf{v}_{i}^{\mathsf{T}}\mathbf{\Sigma}_{\mathbf{x}}\mathbf{v}_{i}$. As $\beta$ increases, the number of rows $n_z$ of $\mathbf{A}_{\beta}$ increases accordingly, thereby prioritizing relevance over compression. Note that, in practice, due to CCA arguments, only $n_z \leq \min \{n_x,n_y\}$ are meaningful, and in this manuscript we will consider $n_z \leq n_y \leq n_x$, without restriction of generality.
Finally, note that the additive noise term $\boldsymbol{\xi}$ is introduced in the original GIB formulation to enable analytical tractability \cite{chechik2003Gaussian}. If included, taking into account $E\{\xi_i^2\}=1$, this would induce a signal-to-distortion ratio on each latent component $z_i$ expressed by $s_i=E\{ \alpha_i^2(\mathbf{v}_i^\mathsf{T}\mathbf{x})^2\}$$=\frac{\beta(1-\lambda_{i})-1}{\lambda_{i}}$, that are ordered in a decreasing fashion, i.e., $s_i \geq s_{i+1}$. However, in practical inference settings, explicitly injecting this noise is typically unnecessary, since $\boldsymbol{\xi}$ is independent of the target variable $\mathbf{y}$, and thus it would not convey any useful information for the learning task \cite{binucci2024opportunistic}.

\subsection{Quantization Noise Model in the Scalar Case}\label{sec:quant_scalar_noise}
The GIB framework naturally provides a supervised linear feature extractor that can be placed as a pre-processing block for subsequent inference. However, in some applications the extracted features $\mathbf{z}$ must be mapped onto a finite-bit representation, before being either stored, or transmitted through a suitable modulation and coding scheme, as it happens for instance in digital GOCs \cite{liu2024ofdm,huang2025d, xie2023robust}, where the bit-rate is constrained by the channel capacity. This calls for explicitly accounting for quantization, both in the features design and the associated bit (rate) allocation, possibly under stringent budget constraints.

Let's start considering the $i$-th (scalar) Gaussian feature $z_i \sim \mathcal{N}(0,s_i)$ at the output of the GIB compression stage, without any noise, i.e., $\xi_i=0$. For a memory-less Gaussian source under mean-squared error (MSE) distortion, classical rate--distortion (RD) theory states that the minimum bit-encoding rate $r_i(D_i)$ required to achieve a distortion $D_i \in (0,s_i]$ is~\cite{cover1999elements}
\begin{equation}\label{eq:rd_gaussian}
    r_i(D_i)=\frac{1}{2}\log_2\!\left(\frac{s_i}{D_i}\right),
\end{equation}
where $r_i$ is expressed in bits/sample. By inversion, the corresponding distortion--rate relationship is
\begin{equation}\label{eq:dr_gaussian}
    D_i(r_i)=s_i 2^{-2r_i}=\frac{\beta(1-\lambda_{i})-1}{\lambda_{i}}2^{-2r_i},
\end{equation}
due to the power loading $s_i=\alpha_i^2$ induced by the GIB on each compressed component.
While \eqref{eq:rd_gaussian}--\eqref{eq:dr_gaussian} describe an asymptotic limit—achievable in principle by long-block vector quantization—they are routinely used as accurate surrogates in the sufficiently high resolution regimes. More specifically, for sufficiently fine scalar quantization, the RD-optimal Shannon bound is approached and the quantization distortion is well captured by an additive 
noise model, whose variance matches the target distortion $D_i$~\cite{gray2002quantization}, as expressed by 
\begin{equation}\label{eq:scalar_quantization_additive_model}
    z_{q_i}=z_i+\eta_{q_i},
\end{equation}
where $\eta_{q_i}\sim\mathcal{N}(0,D_i)$ is independent of $z_i$~\cite{cover1999elements}.
Therefore, in the following we adopt this Guassian-noise additive test-channel approximation, that lets us to deal with analytically tractable expressions, included those for mutual-information. Actually, for entropy-coded dithered quantization (ECDQ) \cite{zamir2002information}, the additive-noise representation is exact and the distortion excess due to non-Gaussian quantization noise is bounded by its divergence from Gaussianity \cite{zamir2002information}, as confirmed by 
simulations in Section {\ref{sec:gaussian_approx_results}}.

Assuming to have a fixed rate constraint, i.e., a finite bit-budget $R_{\text{tot}}$ for the quantized representation $\mathbf{z}_q$, a possible solution is to jointly optimize the number of retained GIB components $n_z$ and the number of bits assigned to each one, according to a minimum distortion criterion.  To this end, we formulate the following optimization problem
\begin{equation}\label{eq:optimal_GIB_quantization}
\begin{aligned}
\min_{n_z,\ \{r_i\}_{i=1,\ldots,n_z}}\quad
&\sum_{i=1}^{n_z}s_i^{(n_z)}2^{-2r_i}+\sum_{i=n_z+1}^{n_y}s_i^{(n_y)}\\
\text{s.t.}\quad
& \text{(a)}\ \sum_{i=1}^{n_z} r_i \leq R_{\mathrm{tot}}, \quad \text{(b)}\ r_i\geq0,\forall i,\\
\end{aligned}
\end{equation}
where $s_i^{(n)}=\frac{\beta_{n+1}^{c}(1-\lambda_{i})-1}{\lambda_{i}}$ is the signal power that each component would receive using totally $n$ components, i.e., using the GIB with $\beta=\beta_{n+1}^c$. The left side of the cost function in \eqref{eq:optimal_GIB_quantization} represents the \emph{quantization} MSE on the $n_z$ retained components of the GIB, while the right side expresses the \emph{representation} MSE, as the power that would be assigned by a full GIB to the last $(n_y-n_z)$ components, which are conversely neglected using just the first $n_z$. For any fixed $n_z$, only the left side of the cost function has to be optimized, and, for non-integer $\{r_i\}\in \mathbb{R}^+$, the solution is known as the reverse water-filling criterion \cite[Th. 13.3.3]{cover1999elements} applied to the $n_z$ GIB-compressed output.
Thus, we may run the optimizer for all the possible $n_z \in [1,n_y]$ and evaluate the best solution in terms of the minimum latent representation distortion, or the informativeness with respect to the final estimation of $\mathbf{y}$.
However, note that this policy, for any fixed $n_z$, can also (further) autonomously select the effectively used components by assigning zero bits, i.e., $r_i=0$, to those less important to minimize the overall distortion of the GIB latent representation.
Thus, we may also think to a simplified procedure that, rather than testing all the possible $\beta$s (e.g., all the possible values for $n_z$), it selects all the possible GIB components, i.e., $n_z=n_y$, and lets the optimization in \eqref{eq:optimal_GIB_quantization} to effectively select those with $r_i>0$. We will further discuss about this option in the simulation results.

\section{Quantization-aware Gaussian Information Bottleneck}\label{sec:quantization}

\begin{figure*}[ht]
    \centering
    \includegraphics[width=0.85\linewidth]{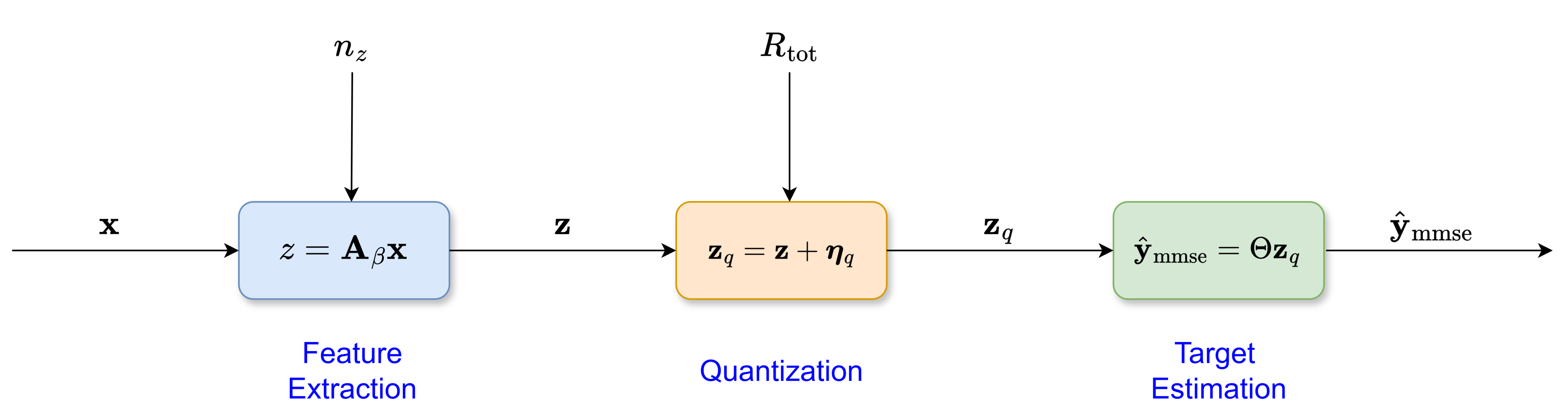}
    \caption{The goal-oriented feature extraction and quantization pipeline. 
The input $\mathbf{x}$ is mapped to a latent representation $\mathbf{z} \in \mathbb{R}^{n_z}$, 
which is quantized under a total rate constraint $R_{\mathrm{tot}}$ to obtain $\mathbf{z}_q$. 
The target $\mathbf{y}$ is then estimated from $\mathbf{z}_q$ via an L-MMSE estimator.}
    \label{fig:system_model}
\end{figure*}

Herein, we exploit the additive Gaussian quantization noise models in \eqref{eq:scalar_quantization_additive_model},
 which preserves the joint Gaussianity between $\mathbf{z}_q$, $\mathbf{x}$, and $\mathbf{y}$. As detailed in the following subsections, this property ensures that the optimality of the linear GIB solution, can be retained also by directly incorporating the quantization in the GIB design, rather than forcing the quantization on top of the GIB solution in \eqref{eq:GIB-CompressionMatrix}.
Thus, rather than minimizing the distortion with respect to the unquantized representation $\mathbf{z}$, as we did in \eqref{eq:optimal_GIB_quantization} by \eqref{eq:dr_gaussian},  we are looking for a quantizer design that directly allocates the available bit budget for the compressed representation $\mathbf{z}_q$ according to the IB objective, i.e., maximizing a trade-off of the mutual information $I(\mathbf{z}_q;\mathbf{y})$ between the quantized features and the target variable, preserving some degree of information $I(\mathbf{z}_q;\mathbf{x})$ with respect to the original $\mathbf{x}$. 

Moreover, the preservation of the GIB structure allows the recovery of the target variable $\mathbf{y}$ via a (linear) minimum mean-squared error (MMSE) estimator, as detailed in \ref{sec:L-MMSE estimator}. We first consider independent SQ for each component of $\mathbf{z}_q$, and then extend the approach to a Gaussian VQ design. The overall feature extraction, quantization, and estimation pipeline is illustrated in Figure~\ref{fig:system_model}.

\subsection{GIB-Aided Scalar Quantization (SQ)}\label{sec:separate_scheme}
Herein we design a GIB-oriented, component-wise \emph{separable} quantization scheme, where the quantization operator factorizes across the entries of the representation. To this end, we consider the following constrained optimization problem
\begin{equation}\label{eq:gib_quantized_problem}
\begin{aligned}
\min_{p(\mathbf{z}_q|\mathbf{x})}\quad
& I(\mathbf{z}_q;\mathbf{x})-\beta I(\mathbf{z}_q;\mathbf{y})\\
\text{s.t.}\quad
& \text{(a)}\ \sum_{i=1}^{n_z} r_i \leq R_{\mathrm{tot}}
\quad\ \text{(b)} \  r_i \geq 0, \ \ \forall i. 
\end{aligned}
\end{equation}
The objective is to determine a stochastic quantization mapping $p(\mathbf{z}_q | \mathbf{x})$ that satisfies the aggregate rate constraint in~(a), while constraint (b) enforces a non-negative rate allocation to each component. 

\begin{remark}[Implicit feature selection under an aggregate rate budget]
The aggregate rate budget constraint~(a) may force the optimal solution to assign zero rate to a subset of components (i.e., $r_i = 0$).
As a consequence, only the components that are most relevant to the task are quantized
while the remaining ones are effectively discarded.
This induces an implicit, task-oriented feature selection mechanism, inherently driven by the available quantization budget $R_{\mathrm{tot}}$.
\end{remark}

By assuming that the quantized 
representation $\mathbf{z}_q$ is described according to the additive distortion model in \eqref{eq:scalar_quantization_additive_model}, we are practically stating  
that (also) $\mathbf{z}_q$ is jointly Gaussian with $\mathbf{x}$, and consequently it can be modelled by the linear regression
\begin{equation}\label{eq:noisy_ib_mapping}
\mathbf{z}_q=\tilde{\mathbf{A}}_{\beta}\mathbf{x}+\boldsymbol{\eta}_q.
\end{equation}
This practically means that any feasible quantizer in \eqref{eq:gib_quantized_problem}, if it adheres to \eqref{eq:scalar_quantization_additive_model}, has to preserve by \eqref{eq:noisy_ib_mapping} the linear structure of the GIB design. However, the difference is that the quantization noise $\boldsymbol{\eta}_q$, plays the role of the stochastic (compression) noise $\boldsymbol{\xi}$, which actually is necessary in \eqref{eq:gib_mapping} only to make the problem analytically tractable, while in practice it can be safely set to zero in GIB compression schemes \cite{binucci2024opportunistic}. Conversely, in our case, where the optimal solution may potentially allocate a different number of bits on each scalar component $z_{q_i}$, the (diagonal) quantization noise covariance matrix $\mathbf{\boldsymbol{\Sigma_\eta}}_q$ depends both on the total number $R_{\text{tot}}$ of quantization bits, and on how many are assigned to each component. Taking this fact into account, we can exploit the classical solution of the GIB, where, as detailed and proved in \cite{chechik2003Gaussian}, a non-identical covariance matrix of the Gaussian noise $\boldsymbol{\eta}_q$ leads to a projection matrix $\tilde{\mathbf{A}}_\beta$ expressed by 
\begin{equation}
\tilde{\mathbf{A}}_\beta=\mathbf{U}^T{\mathbf{D}^{1/2}}\mathbf{A}_{\beta}={\mathbf{D}}^{1/2}\mathbf{A}_{\beta},
\end{equation}
where, (see also  \eqref{eq:vq_noise_cov}) $\mathbf{U}=\mathbf{I}$ are the unitary eigenvectors of the diagonal covariance  matrix $\mathbf{\Sigma}_{\boldsymbol{\eta}_q}=\mathbf{UDU}^T=\mathbf{D} = \mathrm{diag}\left([D_1(r_1),\dots,D_{n_z}(r_{n_z})]\right)$.

This means that the cost function in problem \eqref{eq:gib_quantized_problem} has the same structure of eq. (10) in \cite{chechik2003Gaussian}. Consequently, as detailed in Appendix \ref{sec:appendix_scalar_quantization}, by defining the per-component signal-to-distortion ratio (SDR) as $\rho_i \triangleq s_i/D_i$, it can be proved that the corresponding bit-constrained IB design reduces to the following optimization problem
\begin{equation}\label{eq:go_gib_vq}
\begin{aligned}
&\min_{\beta,\{\rho_i\}_{i=1}^{n_z}}
 \ \ \frac{1}{2}\sum_{i=1}^{n_z}\log_2\!\left(1+\rho_i\right)
 - \beta
 \log_2\!\left(\frac{1+\rho_i}{1+\lambda_i\rho_i}\right) \\
& \hspace{-3pt} \text{s.t.} \ \  
\text{(a)}
\ \frac{1}{2}\sum_{i=1}^{n_z} \log_2(\rho_i)\leq R_{\mathrm{tot}}, \hspace{6pt } 
\text{(b)} 
\  1 \leq \rho_i \leq \rho_i^{\max} \   \forall i,
\end{aligned}
\end{equation}
where 
the optimal rate are $r_i^{\star}=\frac{1}{2}\log_2(\rho_i^{\star})$, and $\rho_i^{\max}$ is a feasibility bound for a unique solution, that is always attained, as detailed in Appendix \ref{sec:convexity_analysis}.
Thus, problem~\eqref{eq:go_gib_vq} determines the set of per-component SDRs, i.e., the inverse of the quantization MSE, that minimize the IB objective while enforcing a constraint $R_{tot}$ on the total rate 

Note that, in problem \eqref{eq:go_gib_vq}, the latent dimension $n_z$ is implicitly fixed by $\beta_{n_z+1}^{c}$, i.e., the 
value of $\beta$ that guarantees $n_z$ retained components through \eqref{eq:GIB-CompressionMatrix}. Therefore, for any fixed $\beta$, we are optimizing the distribution of the total bit-budget $R_{\text{tot}}$ across $n_z(\beta)$ different GIB-like components.
Thus, according to \eqref{eq:go_gib_vq}, at a first sight it seems reasonable to optimize over $\beta$, or equivalently over the number $n_z$ of retained components. However, as already observed in \eqref{eq:optimal_GIB_quantization} and confirmed in the simulation results, in our setting
the best solution is to choose the maximum meaningful $n_z$, e.g., $n_z=n_y$, avoiding the (useless) complexity of an exhaustive search for the best $n_z$, as further detailed in Appendix \ref{sec:mi_maximization_derivations}.
The solution can be obtained via standard constrained optimization arguments, with a Lagrangian expressed by
\begin{equation}
\begin{aligned}
\mathcal{L}\! &\left(  \{\rho_i,\gamma_i,\nu_i\}_{i=1}^{n_z},\eta\right)
=\sum_{i=1}^{n_z} f(\rho_i)
+ \sum_{i=1}^{n_z}\gamma_i\!\left(\rho_i-\rho_i^{\max}\right)\\
&\qquad \quad + \sum_{i=1}^{n_z}\!\nu_i\!\left( 1-\rho_i\right)+ \eta\!\left(\sum_{i=1}^{n_z} r_i(\rho_i)\!-\!R_{\mathrm{tot}}\right),
\end{aligned}
\end{equation}
where $f(\rho_i)$ is the primal objective function of \eqref{eq:go_gib_vq}, while $\gamma_i\geq 0$, $\nu_i\geq 0$, and $\eta\geq 0$ are the dual variables associated with the 
constraints,
As detailed in Appendix~\ref{sec:convexity_analysis}, the non-convex problem \eqref{eq:go_gib_vq} admits an equivalent convex reformulation, which guarantees existence and uniqueness of the solution, as well as its computation.

Specifically, by enforcing the KKT conditions~\cite{boyd2004convex}, complementary slackness implies that, whenever the rate constraint is inactive (i.e., $\sum_{i=1}^{n_z} r_i(\rho_i^{\star})<R_{\mathrm{tot}}$), the corresponding multiplier satisfies $\eta^{\star}=0$. In this case, the optimal per-component SNR admits the closed-form expression
\begin{equation}\label{eq:snr_allocation_rule}
\rho_i^{\star}
=\left[\frac{\beta^{c}_{n_z+1}(1-\lambda_i)-1}{\lambda_i}\right]_{1}^{\rho_i^{\max}},
\end{equation}
where $[x]_{a}^{b}$ =$\min\{\max\{a,x\},b\}$ is a clipping operator in $[a,b]$.
Conversely, when the total-rate constraint is active, i.e., $\sum_{i=1}^{n_z} r_i(\rho_i^{\star})=R_{\mathrm{tot}}$, then $\eta^{\star}\neq 0$ and, as detailed in Appendix \ref{sec:appendix_scalar_quantization}, the optimal SNR is given by
\begin{equation}\label{eq:snr_allocation_rule_tight_constraint}
\rho_i(\eta^\star)
=\left[\frac{-B_i(\eta^{\star})+\sqrt{\Delta_i(\eta^{\star})}}{2A_i(\eta^{\star})}\right]_{1}^{\rho_i^{\max}},
\end{equation}
where $B_i(\eta)\triangleq 1-\beta^{c}_{n_z+1}(1-\lambda_i)+\eta(\lambda_i+1)$,
$A_i(\eta)\triangleq \lambda_i(1+\eta)$, and
$\Delta_i(\eta)\triangleq B_i(\eta)^2-4A_i(\eta)\eta$, with
the optimal multiplier $\eta^\star>0$ that is fixed by satisfying
\begin{equation}
R(\eta^\star)\triangleq \sum_{i=1}^{n_z} r_i\big(\rho_i^\star(\eta^\star)\big)-R_{\mathrm{tot}}=0.
\end{equation}

When $R_{\mathrm{tot}}$ is stringent, the optimal multiplier $\eta^{*}$ enforces the budget by progressively reducing the optimal $\rho_i^{*}$ (hence $r_i^{*}$), potentially driving some components to the minimum-rate boundary, e.g., setting $r_i^*=0$. This implements a budget-driven activation of a subset of features. To practically compute the optimal Lagrange multiplier $\eta^\star$, we introduce the following lemma.

\begin{lemma}\label{lem:monotonicity_lemma}
For every component that is not clipped by the projection on the feasible set, the unconstrained solution $\tilde\rho_i(\eta)$, i.e., the 
argument in rhs of \eqref{eq:snr_allocation_rule_tight_constraint}, is strictly decreasing in $\eta$. Consequently, the projected optimum $\rho_i^\star(\eta)=\big[\tilde\rho_i(\eta)\big]_{1}^{\rho_{\max}}$ is non-increasing in $\eta$ for all $i=1,\dots,n_z$.
\end{lemma}

\begin{proof}
See Appendix~\ref{sec:monotonicity_proof}.
\end{proof}

Since, for the considered Gaussian rate functions, each $r_i(\rho)$ is non-decreasing in $\rho$, Lemma~\ref{lem:monotonicity_lemma} implies that $r_i\big(\rho_i^\star(\eta)\big)$ is non-increasing in $\eta$, and thus $R(\eta)$ is non-increasing in $\eta$. Therefore, whenever the problem is feasible and the constraint is active, the equation $R(\eta)=0$ admits a unique solution $\eta^\star$, which can be efficiently found via the bisection method.

\noindent\textit{Remark:} The optimal design is uniquely characterized by the set of per-component SNRs
$\{\rho_i^\star\}_{i=1}^{n_z}$, with $\rho_i \triangleq s_i/D_i$. Hence, the solution is invariant to any scaling that preserves $\rho_i$. To obtain an implementable encoder/quantizer
pair, one must therefore fix either the component power $\{s_i\}$ or the target distortion levels $\{D_i\}$.
Specifically, if the features powers $\{s_i\}$ are fixed (e.g., by the chosen linear mapping or by a
normalization step), the optimal distortions follow as $D_i^\star = s_i/\rho_i^\star$. Conversely, if the
distortions powers $\{D_i\}$ are fixed by the quantizer design, the required component variances are
$s_i^\star = \rho_i^\star D_i$, which can be enforced by rescaling the representation. In both cases, the
corresponding per-component rate is consistent with~\eqref{eq:rd_gaussian}, i.e.,
$r_i^\star = \frac{1}{2}\log_2\!\big(\rho_i^\star\big)$ (up to the adopted quantization model and any integer-bit
rounding).\footnote{Note that the classical GIB loadings $s_i=\alpha_i^2$ in \eqref{eq:gib_mapping}, derived under unit-noise assumptions, are recovered
by setting $D_i=1$ for all $i$, i.e., by enforcing the same unitary distortion variance on all the components, and setting a sufficiently high $R_{\text{tot}}$, e.g., a bit-budget that would induce a global SNR higher than, or equal to, that one induced by the GIB stochastic mapping in \eqref{eq:GIB-CompressionMatrix}.}
Under this normalization, the GIB transformation is optimal in the sense of~\cite{chechik2003Gaussian}, and the
effect of quantization is entirely captured by the induced SNRs $\{\rho_i^\star\}$.

\subsection{Mutual Information Aided Quantization}\label{sec:mi_aided_quantization}
As already observed in \cite{binucci2024opportunistic}, the use of the GIB for goal-oriented quantization is appealing, as the GIB objective minimizes a rate proxy, namely the mutual information $I(\mathbf{x},\mathbf{z})$, under task-relevant performance constraints. However, two main limitations arise when dealing with quantization: $i)$ the original GIB formulation does not explicitly account for a finite rate budget, since $I(\mathbf{x},\mathbf{z})$ is not directly matched to the actual number of bits available to encode the latent representation; $ii)$ feature selection is governed by the trade-off parameter $\beta$, which determines the relevance of the extracted components through the associated loadings $\{\alpha_i\}_{i=1}^{n_z}$.

From the rate-allocation rules in~\eqref{eq:snr_allocation_rule} and~\eqref{eq:snr_allocation_rule_tight_constraint}, it follows that the resulting rate assignment depends on the GIB-induced representation and, therefore, indirectly on~$\beta$. Consequently, the quantization process may activate fewer features than the maximum number potentially allowed by the overall rate budget~$R_{\mathrm{tot}}$, leading to an inefficient utilization of the available rate resources. This observation motivates the investigation of an alternative bit-allocation strategy that  directly aims at maximizing the mutual information~$I(\mathbf{z}_q,\mathbf{y})$.

As detailed in Appendix~\ref{sec:mi_maximization_derivations}, under multivariate Gaussian assumptions, the transformation maximizing the mutual information~$I(\mathbf{z}_q,\mathbf{y})$ remains linear and still represented by the matrix of left normalized CCA eigenvectors. Hence, also in this case,
the linear transformation structure of the GIB framework is preserved. We therefore consider the following rate-allocation problem
\begin{equation}\label{eq:MI-quantization}
\begin{aligned}
    &\min_{\{\rho_i\}_{i=1}^{n_y}}
    -\frac{1}{2}\sum_{i=1}^{n_y}
    \log_2\left(\frac{1+\rho_i}{1+\lambda_i\rho_i}\right)\\
    &\text{s.t.} \;
    \text{(a)} \;
    \frac{1}{2}\sum_{i=1}^{n_y}\log_2(\rho_i)
    \leq R_{\mathrm{tot}},
    \qquad
    \rho_i \geq 1,\ \forall i.
\end{aligned}
\end{equation}
Specifically, we aim to find a rate allocation that maximizes the mutual information $I(\mathbf{z}_q,\mathbf{y})$ under an aggregated rate constraint (a). Differently from the GIB-aided approach, the number of quantized components is not fixed \emph{a priori}. Instead, the formulation naturally selects the optimal number of active components by assigning zero rate to less relevant dimensions. Said in another way, we are considering the IB optimization problem with $\beta \xrightarrow{} \infty$, without using the proxy $I(\mathbf{x},\mathbf{z}_q)$ to enforce the latent representation entropy, which instead is implicitly fixed by the rate-constraint.  
 Following the derivations reported in Appendix~\ref{sec:mi_maximization_derivations}, it can be shown that the optimal number of bits per component $\{r_i\}$ are determined by the optimal per-component SNRs $\{\rho_i\}$, expressed by 
 \begin{equation}\label{eq:optimal_rho_mi}
\rho_i=\left[\frac{-C_i(\eta)+\sqrt{C_i^2(\eta)-4\eta\lambda_i}}{2\eta\lambda_i}\right]_{1}^{\infty},
\end{equation}
where $\eta \geq 0$ is the Lagrange multiplier associated with the rate constraint, while $C_i(\eta)\triangleq-\frac{(\lambda_i-1+\eta+\eta\lambda_i)}{2\eta\lambda_i}$. From the KKT conditions~\cite{boyd2004convex}, it follows that $\eta^\star=0$ if the rate constraint is inactive. Otherwise, $\eta^\star$ is obtained by solving $
R(\eta)\triangleq \sum_{i=1}^{n_y} r_i\big(\rho_i^\star(\eta)\big)-R_{\mathrm{tot}}=0$, 
which, by similar arguments in the previous subsection, can be efficiently computed via iterative bisection algorithms.

\section{Block Vector Quantization Approaches}\label{sec:block_vector_vq}
So far, we have focused on a separate scalar quantization scheme for the features extracted by the GIB projection.
Interestingly, the aggregate rate constraint (\ref{eq:gib_quantized_problem}a) naturally induces a feature-selection mechanism, as it enables the optimal solution to allocate zero rate to those components that cannot be reliably quantized under a limited rate budget.

While scalar quantization performs generally well when the features are uncorrelated, as it happens with CCA projections, it is however still suboptimal dealing with anisotropic sources, as those induced by the GIB loadings $\alpha_i^2$.
Furthermore, as shown in~\cite{zamir2002information}, a scalar entropy-coded dithered quantizer (ECDQ) operating in the high-rate regime suffers from a constant gap of approximately $0.255$ bits per component with respect to the optimal rate--distortion bound~\cite{cover1999elements}.
Thus, we extend the framework developed in the previous section by introducing a bit-allocation strategy for vector quantization (VQ), that could make better use of the available bits, according to the observations above.

\vspace{-6pt}
\subsection{Noise Model for Vector Quantization}

Let's first consider a single quantizer for the entire feature vector $\mathbf{z}\sim\mathcal{N}(\mathbf{0},\boldsymbol{\Sigma}_{\mathbf{z}})$, i.e., using an $n_z$-dimensional VQ. In the high-rate regime and for well-shaped VQs (e.g., good lattice quantizers), also the VQ output can be modeled by an additive-noise approximation \cite{zamir2002information}
\begin{equation}\label{eq:vq_additive_model}
\mathbf{z}_q=\mathbf{z}+\boldsymbol{\eta}_q,
\end{equation}
where $\boldsymbol{\eta}_q$ is uncorrelated with $\mathbf{z}$, and approximately a zero-mean white Gaussian noise, with covariance
\begin{equation}\label{eq:vq_noise_cov}
    \mathbf{\boldsymbol{\Sigma_\eta}}_q=\mathbb{E}\!\left[\boldsymbol{\eta}_q\boldsymbol{\eta}_q^{\mathsf{T}}\right]
    \approx D_{\mathrm{VQ}}(R)\,\mathbf{I}_{n_z},
\end{equation}
and where $D_{\mathrm{VQ}}(R)$ denotes the (per-component) MSE distortion at rate $R$ (bits/sample). Under high-rate quantization, $D_{\mathrm{VQ}}(R)$ admits the asymptotic characterization\footnote{Differently from \eqref{eq:dr_gaussian}, where the distortion is different on each feature and imposed by its GIB loading $\alpha_i^2$, herein it is the same, and jointly imposed by the geometric mean of all the GIB loadings $\{\alpha_i^2\}$ via $|\mathbf{\Sigma_z}|^{1/n_z}.$}~\cite{gray2002quantization,zamir2002information}
\begin{equation}\label{eq:HR-distortion}
    D_{\mathrm{VQ}}(R)
    \approx
    2\pi e \  G_{n_z}\,\,|\boldsymbol{\Sigma}_{\mathbf{z}}|^{1/n_z}\,2^{-2R/n_z},
\end{equation}
where $G_{n_z}$ is the normalized second moment (NSM) capturing the VQ cell geometry~\cite{gray2002quantization}. For good lattice quantizers, $G_{n_z}\to 1/(2\pi e)$ as $n_z\to\infty$, and \eqref{eq:HR-distortion} asymptotically approaches the Shannon RD limit (applied per dimension) in \eqref{eq:dr_gaussian}~\cite{gray2002quantization,zamir2002information}. 

\vspace{-6pt}
\subsection{Sub-blocks based Vector Quantization}
Also in this case, under a finite rate constraint $R\leq R_{\text{tot}}$, it could make sense to pass to the quantizer the whole features obtained by the GIB with $n_z=n_y$, and let the quantizer compute the best codebook centroids. In practice, due to the high computational complexity to train the VQ codebook (i.e., finding the $2^{R_{\textrm{tot}}}$ centroids $\in \mathbb{R}^{n_z}$), as well as the memory requested to store them, we instead adopt a block-VQ scheme, where VQ is applied in a block-wise fashion on sub-vectors $ \{\mathbf{z}_b\}_{b=1,\ldots,B}$, extracted as a partition of the whole GIB latent representation $\mathbf{z}$.
Specifically, the total rate budget $R_{\mathrm{tot}}$ is distributed across $B$ blocks, each quantized using $R_{\textrm{0}}=\left\lfloor\frac{R_{\textrm{tot}}}{B}\right\rfloor$ bits. \footnote{Using the same rate $R_0$ on each block is instrumental to use the same codebook for all the blocks, as in standard VQ-VAE architectures \cite{van2017neural}. Despite having fewer degrees of freedom, if properly designed, a single VQ still achieves satisfactory performance, as detailed and shown in the simulations.}   
Each feature $z_i$ is assigned to a specific block through the binary assignment variable
\begin{equation}
        \tau_{ib}=\begin{cases}
           1, & z_i \in \mathbf{z}_b, \\
           0, & \text{otherwise}.
        \end{cases}
\end{equation}
Under the high-rate distortion approximation, the distortion associated with the $b$-th block is given by
\begin{equation}\label{eq:high_rate_vq_distortion}
    D_{b}(N_b)= 2\pi e \ G_{n_b}\left|\mathbf{\Sigma}_{\mathbf{z}_b}\right|^{1/N_\textrm{b}}2^{-2R_0/N_\textrm{b}},
\end{equation}
where $\mathbf{\Sigma}_{\mathbf{z}_b}$ denotes the covariance matrix of the features assigned to the $b$-th block, and $N_\textrm{b}=\sum_{i=1}^{n_z}\tau_{ib}$ represents the number of features assigned to the $b$-th block. 

By leveraging the additive Gaussian approximation for the quantization noise and the structure of the optimal GIB mapping~\cite{chechik2003Gaussian}, the resulting optimization problem can be formulated as\footnote{We focus here on MI maximization because, as discussed in the previous sections, it is the best approach in rate-limited regimes. Nevertheless, the same approach can be applied to a GIB-aided block-VQ.}
\begin{equation}
\label{eq:block_vq_gib_final}
\begin{aligned}
\min_{\{\tau_{ib}\}} \quad
& -\frac{1}{2}\sum_{b=1}^{B}\sum_{i=1}^{n_y}
\tau_{ib}\log_2\left(\frac{1+\rho_{ib}}{1+\lambda_i\rho_{ib}}\right) \\[2pt]
\text{s.t.}\quad
& \text{(a)}~~ B\leq  \left\lfloor{\frac{R_{\mathrm{tot}}}{R_0}}\right\rfloor, \hspace{6pt} N_b=\sum_{i=1}^{n_y}\tau_{ib} \\[2pt]
& \text{(b)}~~ \sum_{b=1}^{B}\tau_{ib}=1,\quad
1\leq N_b\leq N_{\max},
\quad \forall i,b \\[2pt]
& \text{(c)}~~ \tau_{ib}\in\{0,1\}, \quad
\rho_{ib}=\frac{s_i}{D_b(N_b)}, \quad \forall i,b .
\end{aligned}
\end{equation}

Problem~\eqref{eq:block_vq_gib_final} together with~\eqref{eq:high_rate_vq_distortion} assigns the compressed features to different quantization blocks in order to maximize the mutual information $I(\mathbf{z}_q,\mathbf{y})$, while satisfying the overall rate budget imposed by constraint~(a). Constraint~(b) ensures that each feature is assigned to a single block, whose size ranges from scalar quantization, i.e., $N_{\textrm{b}}=1$, to the maximum allowed block size, i.e., $N_{\max}\leq N - B +1$. Finally, constraint~(c) restricts the feasible set to be binary, and defines the per-component SDR $\rho_{ib}$.

Problem~\eqref{eq:block_vq_gib_final} is an NP-hard mixed-integer non-linear program. Nevertheless, a low-complexity heuristic solution can be built by first properly permuting the features, which are successively organized in contiguous blocks. Specifically, let
$\boldsymbol{\pi}=(\pi_1,\ldots,\pi_{n_y})$ denote a permutation of
$\{1,\ldots,n_y\}$, and define the reordered feature vector as
\begin{equation}
    \tilde{z}_i = z_{\pi_i}, \qquad i=1,\ldots,n_y .
\end{equation}
Among the $n_y!$ possible permutations of the NP-hard problem, we adopt a comb-like interleaving
rule, defined as
\begin{equation}
    \boldsymbol{\pi}
    =
    \left(
    b+kB
    :
    b=1,\ldots,B,\;
    k=0,\ldots,\left\lfloor \frac{n_y-b}{B} \right\rfloor
    \right).
\end{equation}
This choice is motivated by the structure of the GIB solution, which yields
uncorrelated and anisotropic features ordered by decreasing variance, with
$s_i=\alpha_i^2$. Thus, a consecutive assignment would group components with similar
variances within the same block, resulting in unbalanced statistics among blocks. Conversely, the
proposed interleaved permutation evenly distributes high- and low-variance
components across the blocks, thereby promoting more comparable covariance
profiles and, consequently, inducing similar optimal VQs for all the blocks.\footnote{Although heuristic, this construction is consistent with high-rate VQ theory, which states in \eqref{eq:high_rate_vq_distortion} that the distortion of an $N_b$-dimensional Gaussian vector
scales with $\left|\boldsymbol{\Sigma}_{\mathbf z_b}\right|^{1/N_b}$. Thus,
balancing the component variances across blocks helps preventing strongly
unbalanced covariance determinants and improves the effectiveness of block-wise
quantization.}

The blocks are then formed by partitioning the reordered sequence
$(\tilde{z}_1,\ldots,\tilde{z}_{n_y})$ into $B$ consecutive groups, i.e., adding to \eqref{eq:block_vq_gib_final} the following contiguity constraint 
\begin{equation}
\label{eq:block_contiguity}
\tau_{ib}=\tau_{kb}=1
\ \Longrightarrow\
\tau_{jb}=1,
\quad
\forall\, b,\; 1\leq i<j<k\leq n_y,
\end{equation}
which enforces that any feature $j$ has to be assigned to the block if it is between two other features $i$, and $k$, belonging to the same block. Note that this contiguity constraint implies that all the features have to be used: this makes sense because a single-block VQ potentially would use all of them, exploiting the task-relevant information contained in each one. Thus, the algorithm optimizes the size of each contiguous block, leaving the permutation to design a meaningful preliminary aggregation criterion, which actually is nearly optimal for an equal and fixed size for all the blocks.\footnote{Note however that the optimal block sizes may be (slightly) different from each other. In order to use the same VQ for all the blocks, in the simulations we zero-pad the features of the shortest blocks (if any) to the longest ones.\label{sharedNote}}   

The contiguity constraint in  \eqref{eq:block_contiguity} casts the
feature-to-block assignment as a sequence segmentation problem, which can be solved by 
dynamic programming (DP) \cite{Tatti2013} to produce the optimal block boundaries
\(\{c_b^\star\}_{b=0}^{B}\), with \(c_0^\star=0\) and \(c_B^\star=n_y\).
Thus, the corresponding optimal variables $\{\tau_{ib}^{\star}\}$ are given by
\begin{equation}
\label{eq:tau_from_boundaries}
\tau_{ib}^\star
=
\mathbf{1}\{c_{b-1}^\star < i \le c_b^\star\},
\qquad
i=1,\ldots,n_y,\quad b=1,\ldots,B.
\end{equation}
Finally, the assignment in the original feature domain is obtained through the
inverse permutation:
\begin{equation}
\label{eq:alpha_original_domain}
\alpha_{fb}^\star
=
\tau_{\pi^{-1}(f),b}^\star,
\qquad
f=1,\ldots,n_y,\quad b=1,\ldots,B.
\end{equation}
As detailed in Appendix~\ref{app:dp_partition}, the complexity of the DP algorithm is 
$\mathcal{O}(BN_{\max}n_y)$.

\vspace{-6pt}
\subsection{Variational Approaches for Vector Quantization}\label{sec:variational_vq_approaches}
The quantization schemes considered so far rely on the assumption that both the input and output of the inference task follow a multivariate Gaussian distribution. However, this assumption is often violated in practical scenarios, where the underlying data distributions are generally unknown and closed-form expressions for the mutual information are not available. In such settings, vector-quantized variational autoencoders (VQ-VAEs)~\cite{van2017neural} provide a flexible and effective framework.

VQ-VAEs learn a discrete latent representation of the input data through the joint training of an encoder, a codebook, and a decoder. Specifically, the encoder maps the input $\mathbf{x}$ into a set of continuous latent vectors, which are subsequently quantized by replacing each vector with its nearest neighbor in a finite set of learned embeddings (the codebook). The decoder then reconstructs the desired output from the quantized latent variables.

In their original formulation, VQ-VAEs were designed for reconstruction tasks and are trained using the objective~\cite{van2017neural}
\begin{equation}\label{eq:vq_vae_loss}
    \begin{split}
    \mathcal{L}_{\mathrm{VQ\text{-}VAE}} &=
    \log p(\mathbf{x} | \mathbf{z}_q(\mathbf{x}))
    + \left\| \mathrm{sg}\!\left[\mathbf{z}_e(\mathbf{x})\right]-\mathbf{e} \right\|_2^2 \\
    &\quad + \gamma \left\| \mathbf{z}_e(\mathbf{x})-\mathrm{sg}(\mathbf{e}) \right\|_2^2,
    \end{split}
\end{equation}
where $\mathbf{z}_e(\mathbf{x})$ and $\mathbf{z}_q(\mathbf{x})$ denote the encoder output, and its quantized version obtained through nearest-neighbor lookup in the codebook $\mathbf{e}$, respectively, while $\mathrm{sg}[\cdot]$ represents the stop-gradient operator. The first term corresponds to the reconstruction objective, the second term updates the codebook embeddings towards the encoder outputs, and the third term 
is a commitment regularizer, that encourages the encoder outputs to remain close to the selected codebook vectors~\cite{van2017neural}.

Actually, as it happens in task-oriented communication systems \cite{strinati20216g}, the objective may not be an accurate source reconstruction, but rather the optimization of downstream inference performance. Accordingly, the task relevance should be explicitly accounted for in the objective functions. A natural choice is the variational information bottleneck (VIB) objective~\cite{tucker2022towards}, which can be extended to the vector-quantized setting, yielding the vector-quantized variational information bottleneck (VQ-VIB) formulation
\begin{equation}\label{eq:vib_loss}
    \mathcal{L}_{\mathrm{VIB}} \approx 
    \mathrm{KL}\!\left(q(\mathbf{z}_q | \mathbf{x}) \,\|\, p(\mathbf{z}_q)\right)
    - \beta \,\mathbb{E}_{q(\mathbf{z}_q|\mathbf{x})}\!\left[\log p(\mathbf{y}|\mathbf{z}_q)\right] .
\end{equation}
Here, $q(\mathbf{z}_q|\mathbf{x})$ denotes an approximate (Gaussian) conditional posterior over the quantized latent representation induced by the input $\mathbf{x}$. The KL-divergence term acts as a compression regularizer by constraining the information retained by the latent representation $\mathbf{z}_q$ about the input $\mathbf{x}$, while the second term is a task-oriented objective that promotes representations informative about the target variable $\mathbf{y}$. The latter can also be interpreted as maximizing a variational lower bound on the mutual information $I(\mathbf{z}_q;\mathbf{y})$, with the parameter $\beta$ that controls the trade-off between compression and inference performance, analogously to the GIB framework. 

Actually, as previously discussed and motivated for the Gaussian setting, in rate-limited regimes it makes sense to focus exclusively on maximizing the task-relevant information,
because the compression constraint is implicitly enforced by the available rate budget for the representation. 
Consequently, we propose a variational mutual-information-based vector quantizer (VMI-VQ), that directly promotes the preservation of task-relevant information, by 
the loss
\begin{equation}\label{eq:VMI-VQ-task_oriented_loss}
    \mathcal{L}_{\mathrm{task}} = 
    -\mathbb{E}_{q(\mathbf{z}_q|\mathbf{x})}\left[\log p(\mathbf{y} | \mathbf{z}_q(\mathbf{x}))\right],
\end{equation}
which corresponds to the negative log-likelihood of the target variable given the quantized latent representation. This objective can be interpreted as maximizing a variational lower bound on the mutual information $I(\mathbf{z}_q;\mathbf{y})$~\cite{barber2004algorithm}.

By replacing the (first) reconstruction term in rhs of \eqref{eq:vq_vae_loss} with $\mathcal{L}_{\mathrm{task}}$ (or, alternatively, with $\mathcal{L}_{\mathrm{VIB}}$), the model is encouraged to learn latent representations that are maximally informative about the target variable $\mathbf{y}$, rather than focusing on accurately reconstructing the input $\mathbf{x}$. This way, the resulting quantization scheme becomes inherently task-oriented, aligning the learned representation with the ultimate inference objective.\footnote{Note that the selection of the optimal trade-off parameter $\beta$ in~\eqref{eq:vib_loss} is generally non-trivial, often requiring cross-validation and careful fine-tuning to identify a suitable value~\cite{wu2020learnability}. Conversely, by replacing the reconstruction loss in~\eqref{eq:vq_vae_loss} with the task-oriented loss in~\eqref{eq:VMI-VQ-task_oriented_loss}, no additional hyperparameters are introduced, thereby enabling training to be carried out in a single run without hyperparameter tuning.}

\section{L-MMSE Estimation of the Target Variable}\label{sec:L-MMSE estimator}
The statistical characterization of the quantization noise introduced in the previous sections enables a straightforward computation of the linear minimum mean-squared error (L-MMSE) estimator for reconstructing the target variable $\mathbf{y}$ from the quantized features $\mathbf{z}_q$. In particular, the L-MMSE estimator depends solely on the second-order statistics of the pair $(\mathbf{z}_q, \mathbf{y})$, and is given by \cite{kay1993statistical}
\begin{equation}\label{eq:lmmse}
    \hat{\mathbf{y}} = \mathbf{\Theta}^T \mathbf{z}_q 
    = \mathbf{\Sigma}_{y z_q} \mathbf{\Sigma}_{z_q}^{-1} \mathbf{z}_q.
\end{equation}
Under the assumption of additive quantization noise that is uncorrelated with the input, it follows that $\mathbf{\Sigma}_{y z_q} = \mathbf{\Sigma}_{y z}$ and $\mathbf{\Sigma}_{z_q} = \mathbf{\Sigma}_z + \mathbf{\Sigma}_{\boldsymbol{\eta}_q}$. When the joint statistics of $(\mathbf{x}, \mathbf{y})$ are unknown, replacing the covariance matrices with their sample estimates yields the least-squares (LS) estimator \cite{kay1993statistical}.

Note that the use of the L-MMSE estimator is also motivated from an information-theoretic perspective. Indeed, it is well known that \cite{cover1999elements} 
\begin{equation}\label{eq:ep_inequality}
    \mathbb{E}\big[\|\mathbf{y} - \hat{\mathbf{y}}\|^2\big] 
    \geq \frac{n_y}{2\pi e} \exp\!\big(h(\mathbf{y} | \mathbf{z}_q)\big),
\end{equation}
and by $I(\mathbf{z}_q;\mathbf{y}) = h(\mathbf{y}) - h(\mathbf{y}| \mathbf{z}_q)$, we observe that reducing the estimation error leads to a reduction of the conditional entropy $h(\mathbf{y} | \mathbf{z}_q)$, which in turn increases the mutual information between $\mathbf{y}$ and $\mathbf{z}_q$. Moreover, under jointly Gaussian assumptions, the L-MMSE estimator coincides with the MMSE estimator and achieves the minimum mean-squared error among all estimators \cite{kay1993statistical}, thereby attaining the minimum achievable value of the left-hand side of \eqref{eq:ep_inequality}.

\section{Simulation Results}\label{sec:sim_res}
The proposed quantization approaches are validate herein, by simulations, for generic multivariate regression tasks.
\vspace{-12pt}
\subsection{Tasks 
Description}\label{sec:task_description}
The performance are evaluated on both synthetic and real-world datasets:
\begin{itemize}
    \item \textit{Synthetic dataset:} we consider a multivariate Gaussian regression dataset where $(\mathbf{x},\mathbf{y}) \sim \mathcal{N}(\mathbf{0}, \mathbf{\Sigma}_x, \mathbf{\Sigma}_y, \mathbf{\Sigma}_{xy})$, with $\mathbf{x} \in \mathbb{R}^{n_x}$ and $\mathbf{y} \in \mathbb{R}^{n_y}$. The input vector is generated as $\mathbf{x} \sim \mathcal{N}(\mathbf{0}, \mathbf{\Sigma}_x)$, where $\mathbf{\Sigma}_x$ is a randomly generated positive definite covariance matrix. In our experiments, we set $n_x = 100$ and $n_y = 70$. The output vector $\mathbf{y}$ is obtained according to the linear Gaussian model
\begin{equation}\label{eq:linear_dataset_model}
    \mathbf{y} = \sqrt{\rho}\,\mathbf{H}\mathbf{x} + \sqrt{1-\rho}\,\boldsymbol{\varepsilon},
\end{equation}
where $\rho \in [0,1]$ controls the correlation between $\mathbf{x}$ and $\mathbf{y}$, $\mathbf{H} \in \mathbb{R}^{n_y \times n_x}$ is a random matrix with orthonormal rows, and $\boldsymbol{\varepsilon} \sim \mathcal{N}(\mathbf{0}, \sigma^2 \mathbf{I}_{n_y})$ is an additive Gaussian noise. 
\end{itemize}

To ensure robustness of the experimental evaluation, the (normalized) MSE results obtained on the synthetic dataset are averaged over $100$ independent trials. Each trial consists of $N=\num{1e5}$ samples, generated according to the model in \eqref{eq:linear_dataset_model}, that are globally split into $N_{\mathrm{tr}} = 80{,}000$ training samples and $N_{\mathrm{test}} = 20{,}000$ test samples.

\begin{itemize}
    \item \textit{Human3.6M dataset} \cite{ionescu2014human36m}: we aim at estimating the three-dimensional positions of 17 human skeleton keypoints in the dataset, from their two-dimensional projections on the image plane, resulting in an input dimension of $n_x = 17 \times 2 = 34$ and an output dimension of $n_y = 17 \times 3=51$, corresponding to the 3D coordinates of the keypoints. The dataset is split in $N_{\mathrm{tr}} = 389{,}938$ training samples and $N_{\mathrm{test}} = 135{,}836$ test samples.
\end{itemize}

For the synthetic dataset, the covariance matrices are known, allowing the use of the L-MMSE estimator in \eqref{eq:lmmse} to obtain a reliable estimate of the target variable $\hat{\mathbf{y}}$. In contrast, for the Human3.6M dataset, the target variable is estimated via a least-squares (LS) estimator, which can be interpreted as a sample-based approximation of the L-MMSE solution \cite{kay1993statistical}. Specifically, the matrix $\mathbf{\Theta}$ in \eqref{eq:lmmse} is approximated by
\begin{equation}
    \hat{\mathbf{\Theta}} = \left(\mathbf{Z}_{\mathrm{q,tr}}^{{T}} \mathbf{Z}_{\mathrm{q,tr}}\right)^{-1} \mathbf{Z}_{\mathrm{q,tr}}^{{T}} \mathbf{Y}_{\mathrm{tr}},
\end{equation}
where the design matrix $\mathbf{Z}_{\mathrm{q,tr}} = [\mathbf{z}_{q,1}^{{T}}; \dots; \mathbf{z}_{q,N_{\mathrm{tr}}}^{{T}}] \in \mathbb{R}^{N_{\mathrm{tr}} \times n_z}$ collects the compressed representations of the training samples, while $\mathbf{Y}_{\mathrm{tr}} = [\mathbf{y}_{1}^{{T}}; \dots; \mathbf{y}_{N_{\mathrm{tr}}}^{{T}}] \in \mathbb{R}^{N_{\mathrm{tr}} \times n_y}$ contains the corresponding target outputs. 

\vspace{-12pt}
\subsection{Competitive Approaches}\label{sec:competitive_approaches}
The performance of the proposed quantization schemes are compared against the following baseline approaches:

\begin{itemize}
    \item \textit{\underline{Reverse water-filling rate allocation}:} 
    A rate allocation strategy based on reverse water-filling, which distributes bits across the components so as to minimize the MSE between the standard GIB representation $\mathbf{z}$ and its quantized version $\mathbf{z}_q$. This baseline corresponds to the classical optimal rate allocation strategy for Gaussian sources under an MSE distortion criterion \cite{cover1999elements}. Hence, this is a natural benchmark to assess the gains achievable by the task-aware structure induced by the GIB/MI representation.

    \item {\textit{\underline{Entropy-proportional rate allocation}}:}
    A heuristic allocation rule where, as envisioned in \cite{binucci2024opportunistic}, the bits assigned to each feature are proportional to its marginal entropy. This baseline reflects a simple yet commonly used strategy in practical compression systems, where components with larger variability are allocated more bits. As such, it serves as a lightweight data-driven alternative that does not explicitly exploit the task objective.

 \item {\textit{\underline{VMI-VQ}}:} the variational approach proposed in Section~\ref{sec:variational_vq_approaches} where the NN encoder and quantizer are jointly trained to maximize task-relevant information. This baseline constitutes a flexible, black-box method, capable of learning task-dependent representations directly from data, both in Gaussian and non-Gaussian settings, and provides a natural benchmark for comparison with analytically derived quantization strategies.
\end{itemize}

\vspace{-6pt}
\subsection{Additive Approximation of the Quantization Noise}\label{sec:gaussian_approx_results}
The goal of this section is to validate the additive noise approximation for the quantization error introduced in Section~\ref{sec:quantization}, considering both scalar and vector quantization schemes. For the scalar case, we employ an entropy-coded dithered quantizer (ECDQ) \footnote{ECDQ is implemented using subtractive uniform dithering generated from a deterministic seed shared by encoder and decoder. For each component, a dither signal \(u \sim \mathcal U(-\Delta/2,\Delta/2)\) is generated identically at both sides. The encoder computes the quantization index as \(k=\mathrm{round}((x+u)/\Delta)\), while the decoder reconstructs \(\hat{x}=\Delta k-u\). The quantization step \(\Delta\) is numerically selected for each component so as to achieve the distortion corresponding to the target rate according to~\eqref{eq:dr_gaussian}.}, in accordance with the additive noise model described in Section~\ref{sec:quant_scalar_noise}. For the vector case, we consider block-vector quantization based on lattice codebooks with maximum block size $N_{\textrm{F}}=4$, where quantization is performed using the $D_4$ lattice \cite{conway2013sphere}.

The experimental analysis is conducted on a synthetic multivariate Gaussian dataset, according to the model described in Section~\ref{sec:task_description}, with correlation coefficient $\rho=0.8$. The number of compressed features is fixed to $n_z=40$. We evaluate both the empirical and theoretical reconstruction (normalized) mean-squared error (NMSE) of the latent representation $\mathbf{z}$, as well as the corresponding MSE of the target variable $\mathbf{y}$.

Figure~\ref{fig:reconstruction_nmse_vs_rate_budget_scalar_vector} shows the reconstruction NMSE with respect to the relevance variable $\mathbf{z}$ for both scalar and vector quantization schemes. Dashed lines correspond to the theoretical rate--distortion lower bound under ideal quantization, whereas solid lines report the empirical NMSE achieved by the considered lattice quantizers, evaluated on the test set. As expected, the NMSE decreases with increasing rate for all schemes.
From Figure~\ref{fig:reconstruction_nmse_vs_rate_budget_scalar_vector}, it can be observed that VQ yields a non-negligible performance improvement compared to scalar quantization. This gain mostly depends on the fact that, even though the GIB-extracted features are independent, they are characterized by an anisotropic nature, e.g., a different variance, as discussed in Section~\ref{sec:gaussian_ib_section}.
Furthermore, the observed NMSE gain is also (slightly) due to the smaller shaping factor of the $D_4$ lattice, whose normalized second moment is $G(D_4)\approx 0.08512$, compared to the scalar case where $G_{\mathrm{scalar}} = 1/12$.
Finally, it is possible to observe that the NMSE gap, between theoretical predictions and empirical results, progressively diminishes for increasing rates, approaching the asymptotic limits $\Delta_s = 10\log_{10}\!\left(\frac{\pi e}{6}\right)\approx 1.53\,\mathrm{dB}$ and $\Delta_{D_4} = 10\log_{10}\!\big(G(D_4)(2\pi e)\big)\approx 1.16\,\mathrm{dB}$, for scalar and vector quantization, respectively \cite{gray2002quantization}.

\begin{figure}[ht]
    \centering
    \includegraphics[width=0.85\linewidth]{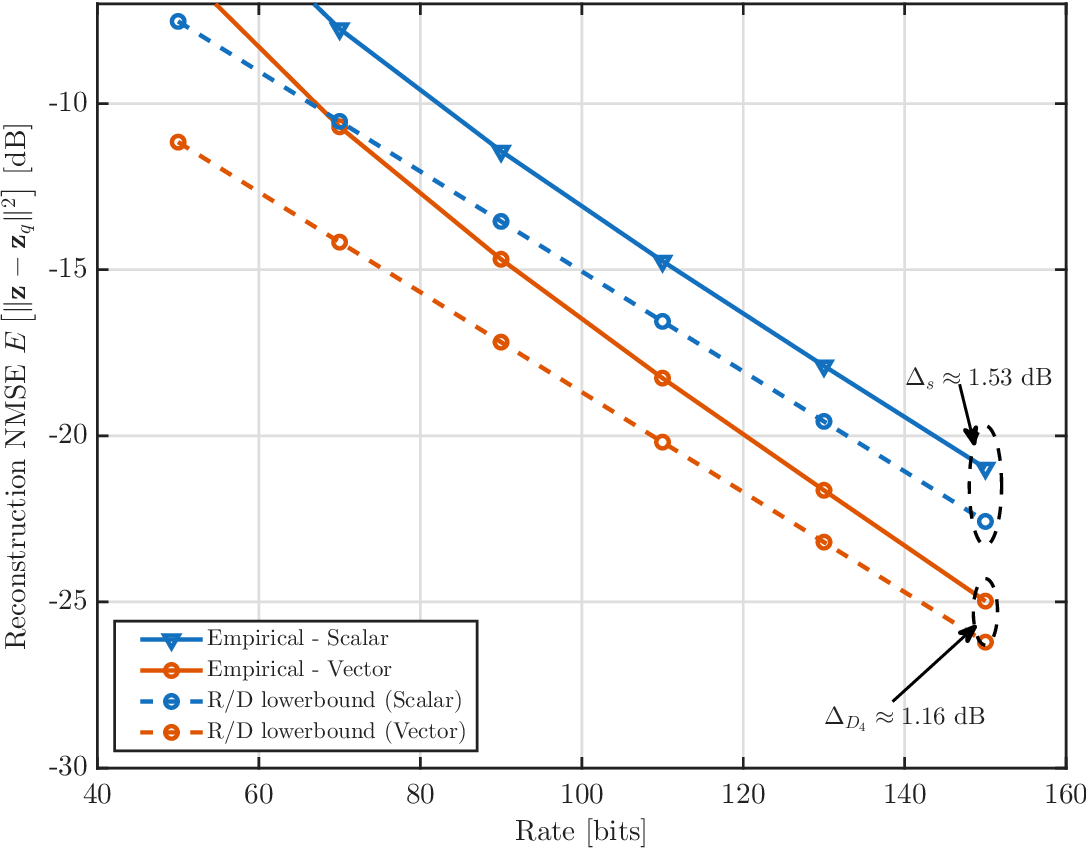}
    \caption{Reconstr. NMSE 
    for vector and scalar quantization.}
    \label{fig:reconstruction_nmse_vs_rate_budget_scalar_vector}
\end{figure}

\begin{figure}[ht]
    \centering
    \includegraphics[width=0.85\linewidth]{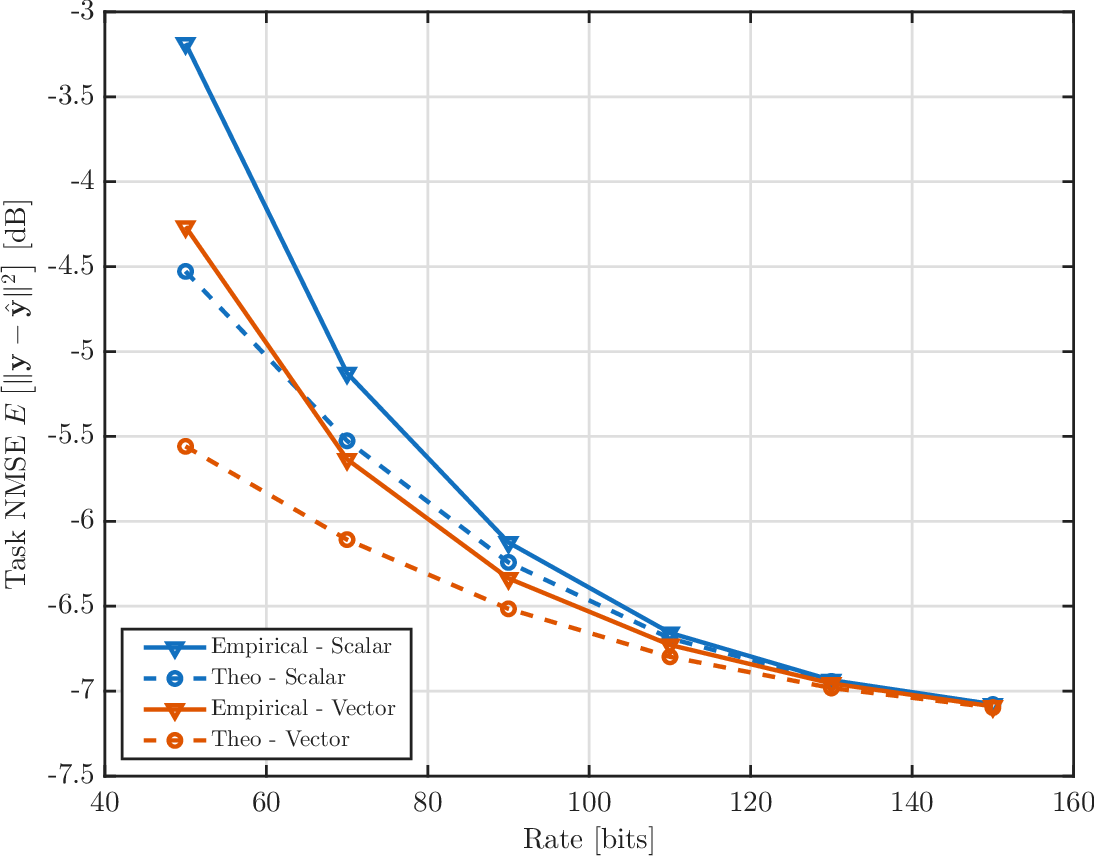}
    \caption{Task NMSE
    for vector and scalar quantization.}
\label{fig:task_nmse_vs_rate_budget_scalar_vector}
\end{figure}

Figure~\ref{fig:task_nmse_vs_rate_budget_scalar_vector} reports the task NMSE achieved by L-MMSE estimation of the target variable $\mathbf{y}$. As expected, the VQ is capable to turn the \emph{representation} NMSE advantage with respect to scalar quantization, in a clear improvement also in \emph{task} NMSE, especially in low-rate regimes.

\begin{figure}[ht]
    \centering
    \includegraphics[width=0.85\linewidth]{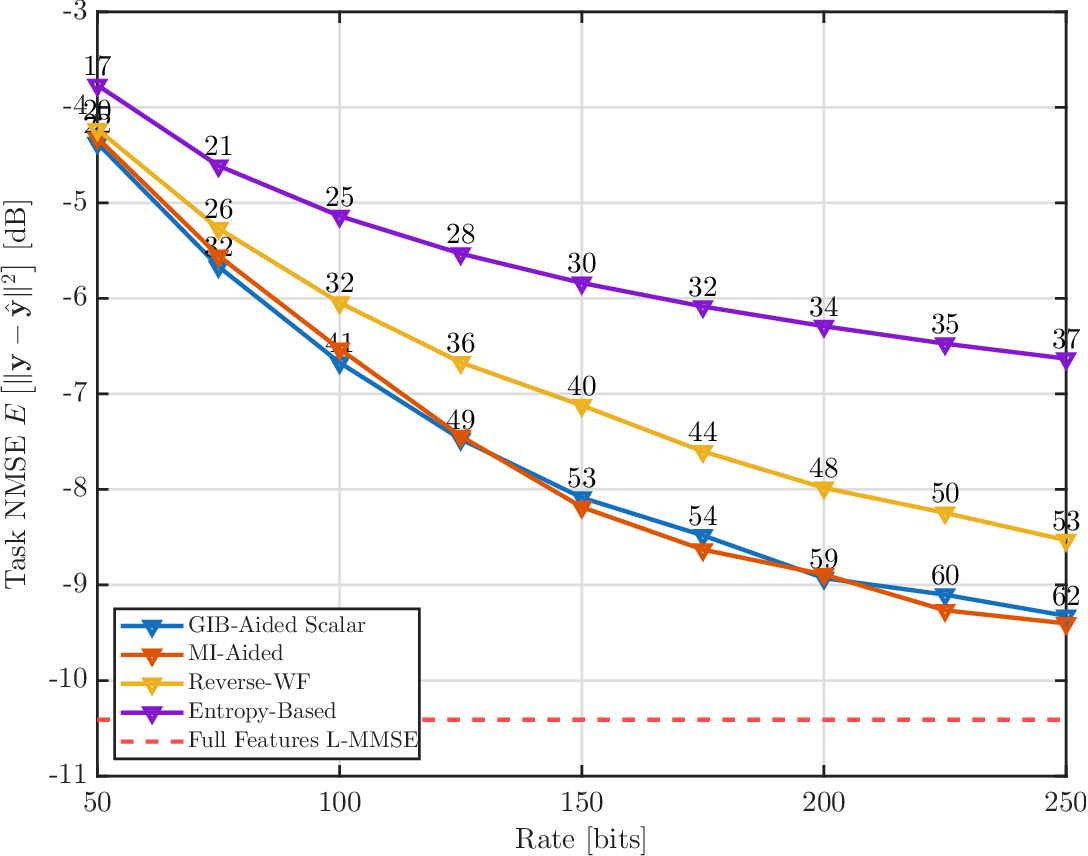}
    \caption{Task NMSE $\mathbb{E}\left[\|\mathbf{y}-\hat{\mathbf{y}}\|^2\right]$: numbers on the markers show how many features are effectively quantized ($r_i^*>0$). 
    }
    \label{fig:NMSE_versus_rate_budget}
\end{figure}

\vspace{-6pt}
\subsection{Separate Scalar Quantization (SQ)}

\begin{figure*}[ht]
    \centering
    \includegraphics[width=0.85\linewidth]{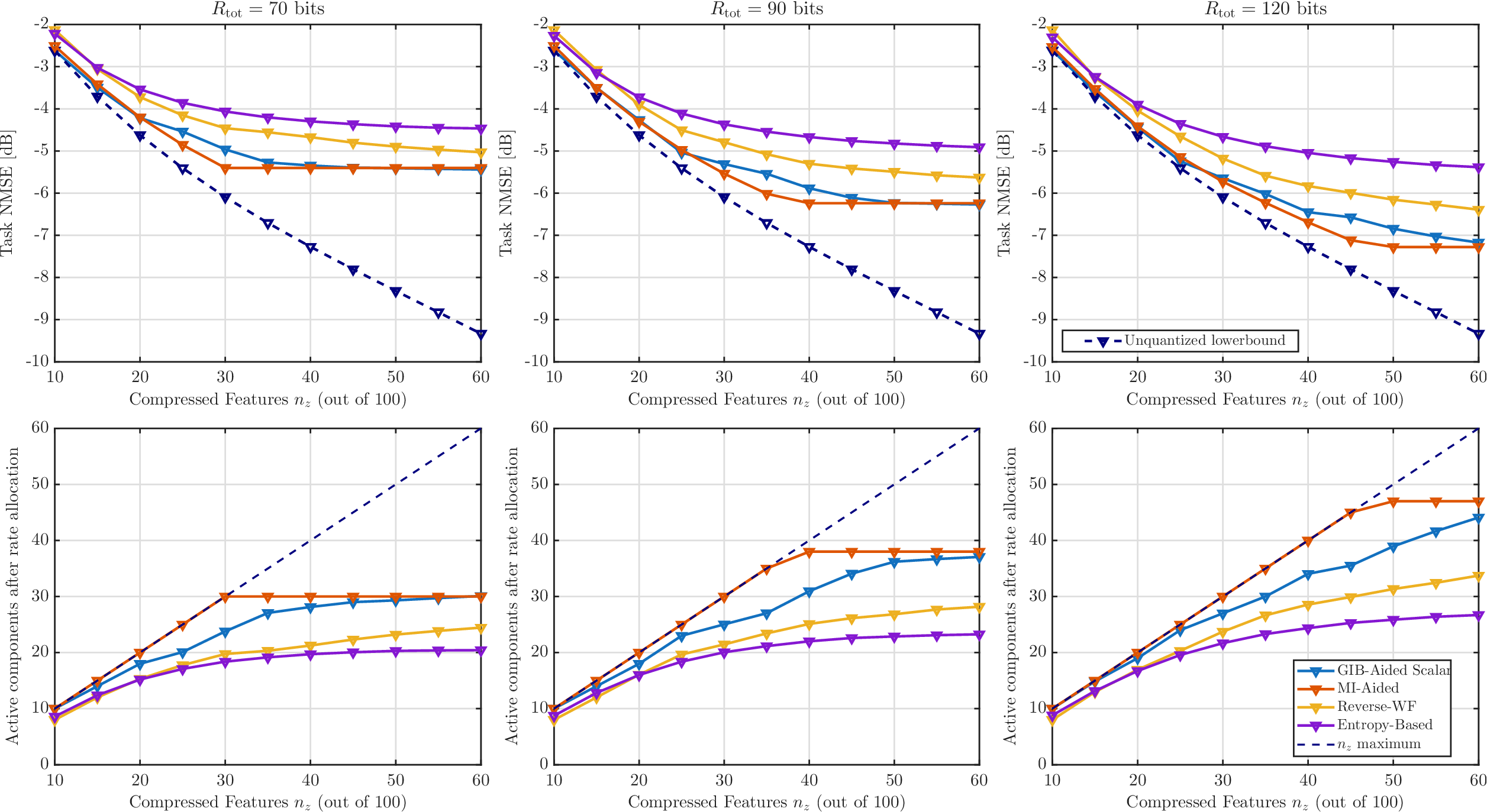}
    \caption{Task NMSE $\mathbb{E}\!\left[\|\mathbf{y}-\hat{\mathbf{y}}\|^2\right]$ (top row) and number of effectively selected components (bottom row) as a function of the number of compressed features $n_z$, under different rate budgets $R \in \{70,90,120\}$ bits.}
    \label{fig:nmse_versus_compressed_features}
\end{figure*}

This section assesses the performance of the proposed GIB/MI-aided quantization schemes both on synthetic and real datasets, comparing them with the baselines described in Section~\ref{sec:competitive_approaches}. In the first test, we consider a synthetic multivariate regression dataset generated according to the model introduced in Section~\ref{sec:task_description}, setting a  correlation coefficient $\rho = 0.8$.

Figure~\ref{fig:NMSE_versus_rate_budget} reports the average \emph{task} NMSE versus the total rate budget $R_{\text{tot}}$ for all the competing approaches, with the maximum number of compressed components set to $n_z = n_y = 70$. Specifically, (all) the features are quantized with the \emph{total} rate constraint $R_{\mathrm{tot}} \in \{50,75,100,125\dots,250\}$ bits, i.e., an average number of bits per component roughly $\in [ 0.71, 3.57]$. 
Under this setting, the number of components that are effectively quantized (shown by labels on each simulation marker) is determined by the solutions of the optimization algorithms (e.g., setting $r_i^*=0$ for the unused ones) that, this way, optimally exploit the available rate budget.

As expected, both the (full-feature) GIB-aided and the MI-aided quantization schemes consistently outperform the reverse-WF and entropy-based strategies, across all the considered rate budgets, activating more features with respect to the "variance/entropy" based quantization schemes. Moreover, the MI-aided scheme exhibits a performance nearly identical to that of the GIB-aided approach. This behavior can be interpreted by observing that fixing the number of components to $n_z = n_y$ is equivalent to maximize the hyperparameter $\beta$ in \eqref{eq:ib_prob}, thereby driving the GIB objective toward the maximization of the mutual information $I(\mathbf{z}, \mathbf{y})$. 

In addition, the performance gap between the GIB-aided method and reverse-WF widens as the available rate increases. This trend can be explained by the structure of the reverse-WF cost function in \eqref{eq:optimal_GIB_quantization}: as the rate budget grows, the reverse-WF solution increasingly emphasizes reconstruction accuracy, leading to a rate allocation that is suboptimal with respect to task-relevant components.

Figure~\ref{fig:nmse_versus_compressed_features} shows the empirical \emph{task} NMSE versus the \emph{nominal} number of compressed features, e.g., as a function of the increasing values $n_z<n_y$. Note that, while $n_z<n_y$ is fixed \emph{a priori} for the other schemes, this would not be the case for the MI formulation, where $n_z=n_y$ by design in \eqref{eq:MI-quantization}: however, for comparison fairness and sake of investigation, herein we exploited only the first $n_z<n_y$ quantized features, among the $n_y$ obtained by the optimal MI (quantized) solution.
The results are plot for three different rate budgets, namely $R_{\mathrm{tot}} \in \{70, 90, 120\}$ bits. Also in this case, the MI-aided strategy achieves the lowest NMSE, confirming that it's not recommended to fix a-priori the number $n_z$ of latent features.

\begin{figure}[ht]
    \centering
    \includegraphics[width=0.85\linewidth]{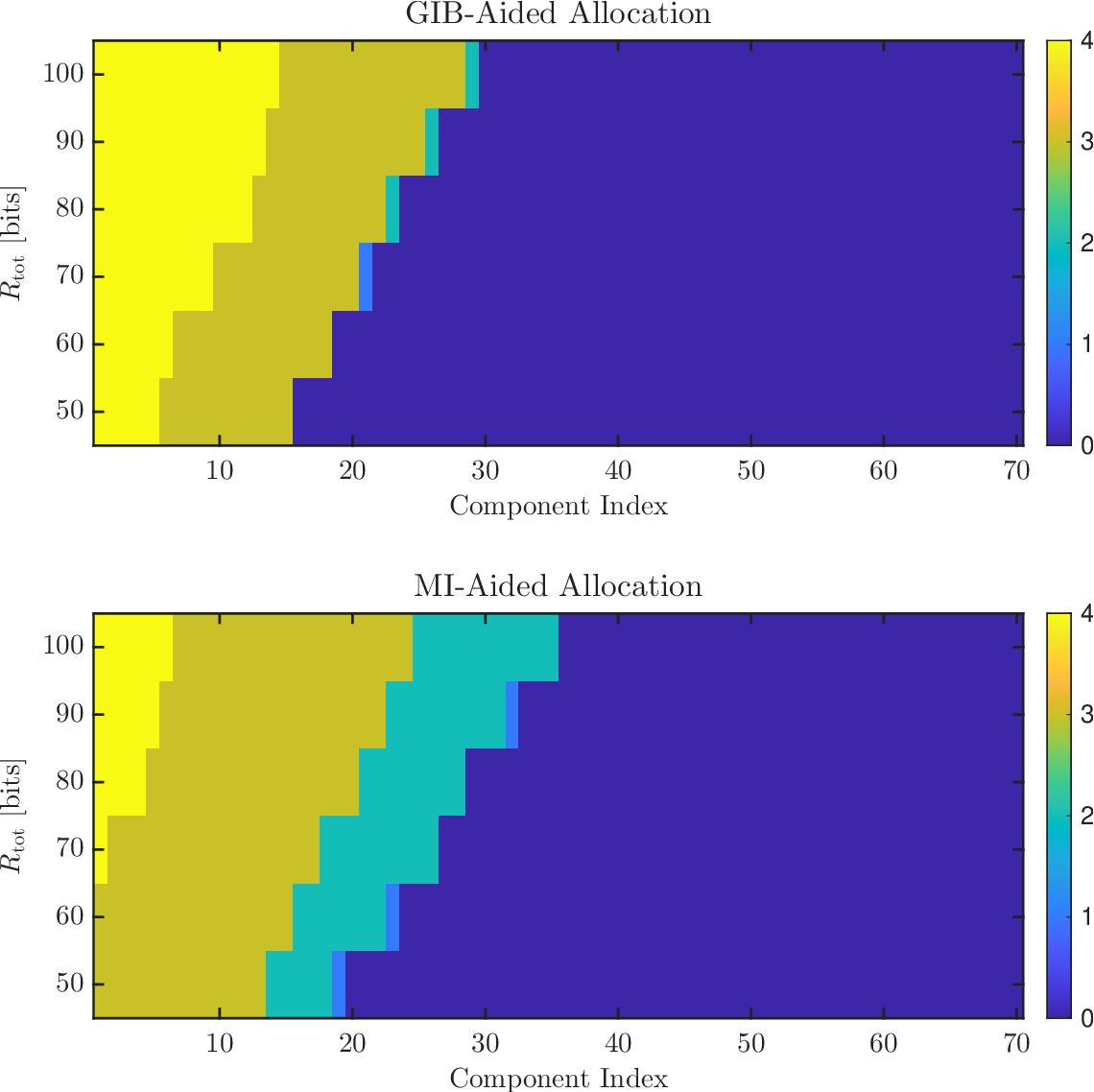}
    \caption{Rate distribution among different components for the GIB-Aided, and the MI-Aided strategies.}
    \label{fig:rate_allocation_heatmap}
\end{figure}

The bottom panel shows the average effective number of selected components for the three strategies as the \emph{nominal} number of compressed features $n_z$ increases. It can be observed that the MI-aided quantization scheme retains a larger number of active components (e.g., with $r_i^{\star}>0$) compared to the competing approaches, thereby achieving a more favorable trade-off between compression and task relevance.

To further support this observation, Figure~\ref{fig:rate_allocation_heatmap} illustrates the rate allocation across components obtained with both MI-aided and GIB-aided schemes. In particular, we consider a variable total rate budget $R_{\mathrm{tot}} \in \{50, 60, \dots, 100\}$ bits. For each budget value, the MI-based allocation is first used to determine the optimal number of active components $n_z^{*}$. Then, the GIB-aided allocation is tuned to the MI, by selecting the minimum value of $\beta$ that yields the same number of active components $n_z^{*}$.
The performance gap can be attributed to the different allocation profiles. The MI-aided strategy selects a larger set of components and distributes the available rate more evenly among them, as illustrated in Figure~\ref{fig:rate_allocation_heatmap}. In contrast, the GIB-aided scheme tends to concentrate the rate on a smaller subset of components, effectively over-quantizing them while neglecting others, which ultimately results in degraded (task) inference performance.

\vspace{-6pt}
\subsection{Vector Quantization (VQ)}
\begin{figure}[ht]
    \centering
    \includegraphics[width=0.85\linewidth]{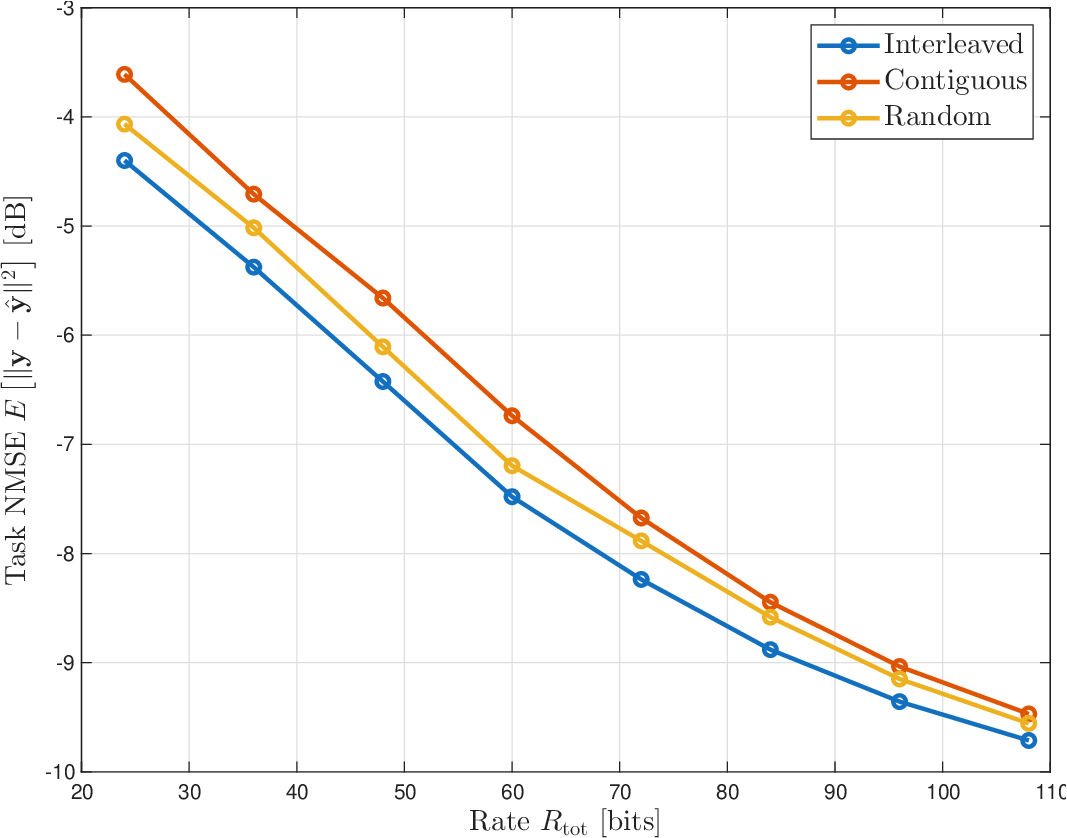}
    \caption{Task NMSE $\mathbb{E}\left[\|\mathbf{y}-\hat{\mathbf{y}}\|\right]$ vs. rate for the the interleaved, the contiguous, and a random feature permutation.
    }

\label{fig:task_nmse_vs_rate_budget_permutation}
\end{figure}
\begin{figure}[ht]
    \centering
    \includegraphics[width=0.85\linewidth]{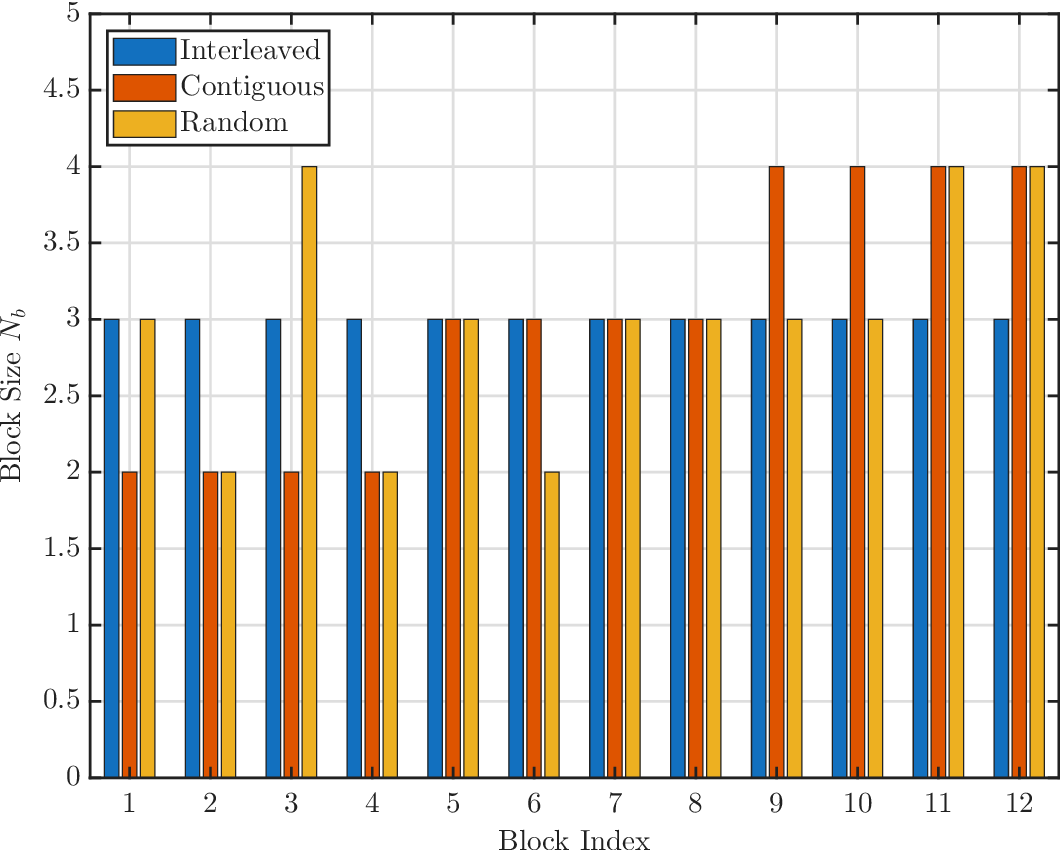}
    \caption{Average block size for the three feature permutation schemes ($R_{\textrm{tot}}=108$ bits).}
    \label{fig:block_size_versus_permutation}
\end{figure}

\begin{figure*}[t]
    \centering    \includegraphics[width=0.75\linewidth]{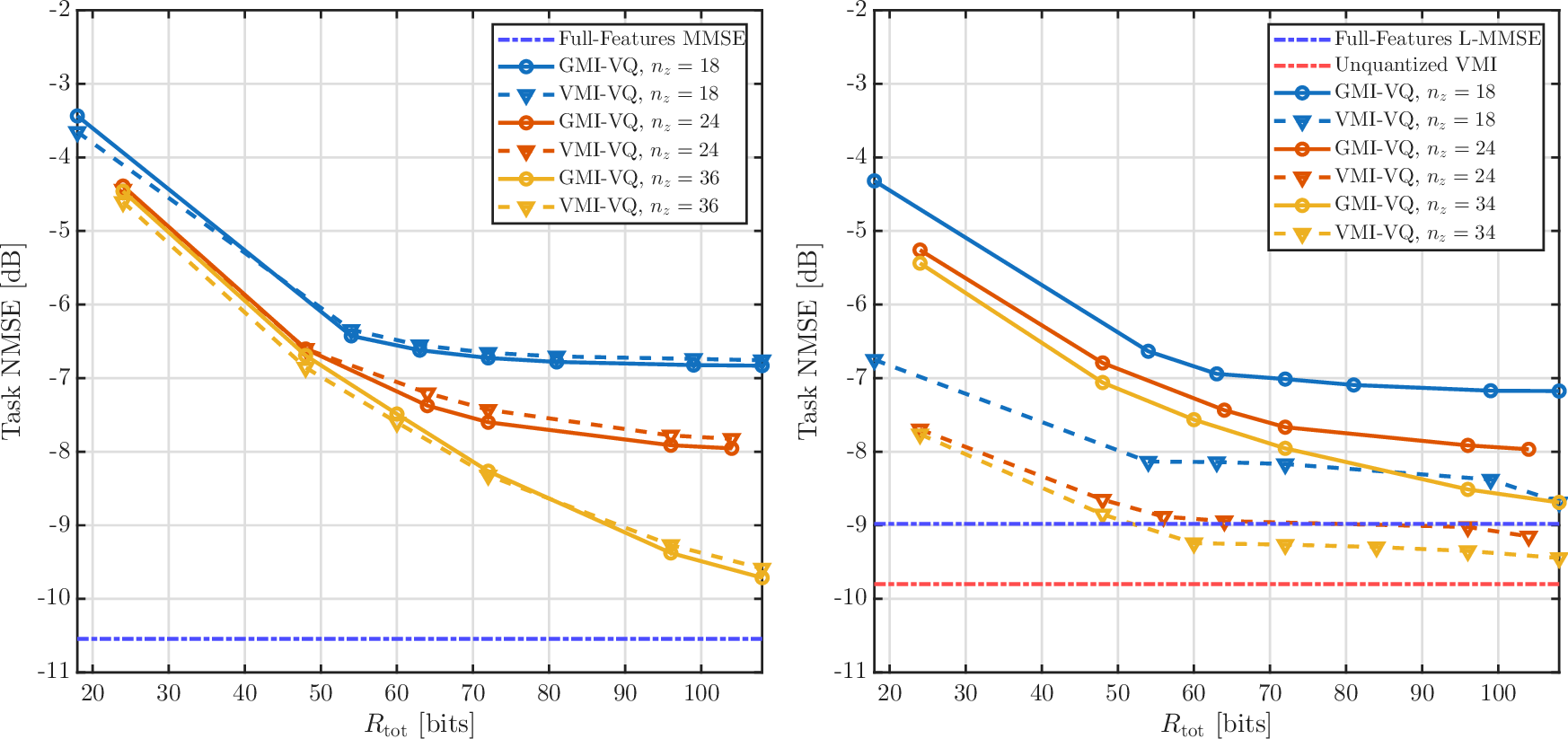}
    \caption{Left panel: Task NMSE as a function of the rate budget for GMI-VQ, and VMI-VQ evaluated on a Gaussian synthetic dataset. Right panel: Corresponding simulation results on the h36m dataset.}
    \label{fig:vq_vae_comparisons}
\end{figure*}

This section evaluates performance of the MI-aided VQ schemes on synthetic and real-world datasets. We compare the Gaussian Mutual Information VQ (GMI-VQ) in \eqref{eq:block_vq_gib_final} with the Variational Mutual Information VQ (VMI-VQ) in \eqref{eq:vq_vae_loss}, as modified by \eqref{eq:VMI-VQ-task_oriented_loss}.
Experiments are conducted on a synthetic multivariate Gaussian dataset generated according to \eqref{eq:linear_dataset_model}, with $n_x=50$ and $n_y=36$, and on the Human3.6M dataset.

For fairness, the VMI-VQ encoder is designed with a complexity comparable to that of the linear GMI-VQ encoder. Specifically, we use a shallow three-layer neural network, with ReLU activations that maps $\mathbf{x}\in\mathbb{R}^{n_x}$ to a latent representation composed of $B$ tokens, each of dimension $N_{\textrm{F}}$, so that $n_z=BN_{\textrm{F}}$. Each token is quantized by nearest-neighbor assignment using a shared codebook of size $K$, yielding a fixed-rate representation with total rate $R_{\mathrm{tot}}=B\log_2 K$.

For both datasets, we evaluate two compressed latent dimensions, $n_z\in\{18,24\}$, as well as the full-feature representation with $n_z=\min(n_x,n_y)$. This corresponds to $n_z=36$ for the synthetic Gaussian dataset and $n_z=34$ for Human3.6M. For VMI-VQ, we set $N_{\textrm{F}}=3$ when $n_z\in\{24,36\}$ and $N_{\textrm{F}}=2$ when $n_z=18$, in order to keep the shared codebook size moderate and reduce training complexity. The VMI-VQ models are trained as described in Section~\ref{sec:variational_vq_approaches}, replacing the conventional reconstruction loss with the task-oriented log-likelihood objective in \eqref{eq:VMI-VQ-task_oriented_loss} for $5{,}000$ epochs, with a batch size $2048$, and setting a  rate $l_r=10^{-2}$.

For GMI-VQ, we use $B \in \{9, 8, 12\}$ feature blocks and vary the per-block rate as $R_0\in\{2,3,\ldots,9\}$ bits/block, resulting in a total rate $R_{\textrm{tot}}=BR_0 \in \{18,\dots,108\}$ bits. For each feature block, a VQ is trained on the training set, using a single shared codebook across blocks.\footref{sharedNote} 

First, we study the effect of feature-permutation on the task performance. We use a synthetic multivariate Gaussian dataset generated as in Section~\ref{sec:task_description}, with $n_x=50$ and $n_y=36$, using $N_{\textrm{tr}}=80{,}000$ training samples and $N_{\textrm{test}}=20{,}000$ test samples. We compare three permutations: the contiguous ordering induced by the GIB, the interleaved scheme introduced in Section~\ref{sec:block_vector_vq}, and a random assignment of features across blocks. The quantizer operates on the full feature representation (i.e., $n_z=n_y=36$).

Figure~\ref{fig:task_nmse_vs_rate_budget_permutation} shows the average NMSE on the test-set versus $R_{\textrm{tot}}$. The interleaved permutation provides the best performance by forming blocks with more homogeneous distribution of the feature variances, thereby mitigating the loss induced by using a shared VQ codebook across all blocks. Conversely, sorting features by decreasing variance and applying a contiguous ordering creates high unbalance among blocks, leading to a noticeable performance degradation. Consistently with this interpretation, the random permutation lies between these two extremes: by mixing features, it partially homogenizes the blocks and therefore achieves intermediate performance.

Figure \ref{fig:block_size_versus_permutation} reports the block sizes obtained by the three schemes for $R_{\textrm{tot}}=108$ bits. As expected, the contiguous permutation allocates smaller blocks to the most task-relevant components and larger blocks to less relevant ones, since smaller blocks reduce the average quantization distortion with \eqref{eq:HR-distortion}. In contrast, random and interleaved permutations spread feature variances more evenly across blocks, yielding more homogeneous block sizes.

Figure~\ref{fig:vq_vae_comparisons} compares VMI-VQ and GMI-VQ across compression levels and rate budgets. On the synthetic Gaussian dataset, VMI-VQ attains performance comparable to GMI-VQ and slightly outperforms it at low rates, where the high-rate quantization-noise approximation becomes inaccurate. In contrast, at higher rates, GMI-VQ benefits from the closed-form GIB mapping, which maximizes the mutual information with $\mathbf{y}$, whereas VMI-VQ relies on more complex learned models that may yield coarser approximations of the optimal compression map. For each rate budget $R_{\textrm{tot}}$, performance improves with $n_z$, approaching the full-features unquantized L-MMSE estimator, which is MMSE-optimal under multivariate Gaussian assumptions~\cite{kay1993statistical}. This indicates that, in rate-limited regimes, retaining the full feature dimension before quantization is preferable to the dimensionality reduction typically adopted in conventional IB-based methods.

Conversely, on Human3.6M, VMI-VQ significantly outperforms GMI-VQ. This gap is mainly due to the mismatch between the Gaussian assumptions of the linear GMI-VQ framework and the nonlinear structure of the inference task. This interpretation is supported by the fact that, on the full feature set, the unquantized L-MMSE estimator performs worse than the unquantized encoder--decoder architecture. In non-Gaussian regimes, nonlinear embeddings learned by neural encoders therefore provide a clear advantage over linear feature extraction.

\section{Conclusion and Future Work}\label{sec:conclusion}
This paper proposes optimal designs for scalar and vector quantization on latent representations extracted via the Gaussian Information Bottleneck (GIB). Simulation results confirm that, under finite rate budget constraints, the quantization design should directly target mutual information rather than original IB optimization.
These findings can be extended beyond the Gaussian setting by leveraging deep learning approaches based on variational quantized encoders, which optimize a tractable lower bound on the mutual information.
Future research directions include the analysis of a wider class of inference tasks, as well as the integration of task-oriented quantization strategies into edge inference systems, jointly accounting for digital modulation and channel coding under practical communication constraints.

\begin{appendices}
\section{Derivation of the Scalar Quantization Problem}~\label{sec:appendix_scalar_quantization}
We observe that the linear GIB mapping in \eqref{eq:gib_mapping}-\eqref{eq:GIB-CompressionMatrix}     emerged as the solution of the IB problem under Gaussian assumptions. Conversely, in this case we are imposing the linear structure for $\mathbf{z}_q$ in \eqref{eq:noisy_ib_mapping} according to optimal quantization design arguments, in high rate regimes. Thus, the cost function in problem \eqref{eq:gib_quantized_problem}
has exactly the same structure of Eq.~(6) in~\cite{chechik2003Gaussian}, and consequently, the compression matrix $\tilde{\mathbf{A}}_{\beta}$ is necessarily equivalent to the optimal one for the GIB, i.e., constructed from the eigenpairs $\{\lambda_i,\mathbf{v}_i\}_{i=1}^{n_x}$ of the CCA matrix $\mathbf{\Sigma}_{\mathbf{x}|\mathbf{y}}\mathbf{\Sigma}_{\mathbf{x}}^{-1}$. However, in this case the corresponding loadings $\{\alpha_i\}$ have to be selected as a function of the per-component rate $r_i$, i.e., of the distortion $D_i$. Thus, employing the closed form expressions for the mutual information under jointly Gaussian models, exploiting the same derivation that lead to Eq.~(26) in~\cite{chechik2003Gaussian}, problem~\eqref{eq:go_gib_vq} can be equivalently rewritten as
\begin{equation}\label{eq:reduced_quant_prob}
\begin{aligned}
\min_{\{s_i,D_i\}_{i=1}^{n_z}}\quad
& \frac{1}{2}\sum_{i=1}^{n_z}\log_2\!\left(\frac{s_i+D_i}{D_i}\right)
-\beta_{n_z}^{c}\log_2\!\left(\frac{s_i+D_i}{D_i+\lambda_i s_i}\right) \\
\text{s.t.}\quad
& \text{(a)}\ \sum_{i=1}^{n_z} r_i\!\left(\frac{s_i}{D_i}\right)\leq R_{\mathrm{tot}},\\
& \text{(b)}\ 0< D_i \leq s_i, \quad \forall i=1,\dots,n_z,
\end{aligned}
\end{equation}
where we defined $s_i \triangleq \alpha_i^2$. 
Interestingly, both the objective function and the constraints in~\eqref{eq:reduced_quant_prob} depend on $(s_i,D_i)$ only through the per-component SNR $\rho_i \triangleq s_i/D_i$, that substituted into~\eqref{eq:reduced_quant_prob} lead to~\eqref{eq:go_gib_vq}.

\section{Convexity Analysis of the Scalar Quantization Problem}\label{sec:convexity_analysis}
The rate-allocation problem formulated in Section~\ref{sec:separate_scheme} is nonconvex. In particular, the objective function is the sum of quasi-convex (yet nonconvex) terms, and thus it is not even quasi-convex. Moreover, the maximum-rate constraint in~(\ref{eq:go_gib_vq}-a) is not convex, since it can be expressed as a sublevel set of concave functions. Nevertheless, by means of a simple change of variables, the problem can be recast into a convex program, which guarantees existence and uniqueness of the optimal solution.

Define $\delta_i \triangleq \ln(\rho_i)$, i.e., $\rho_i = 2^{\delta_i}$. The problem can be rewritten as
\begin{equation}\label{eq:convexified_problem}
\begin{split}
    \min_{\{\delta_i\}_{i=1}^{n_z}}\hspace{10pt}
    &
    \sum_{i=1}^{n_z}
    \log_2\!\bigl(1+2^{\delta_i}\bigr)
    +\frac{\beta^{c}_{n_z}}{1-\beta^{c}_{n_z}}
    \sum_{i=1}^{n_z}\log_2\!\bigl(1+\lambda_i 2^{\delta_i}\bigr)\\
    \text{s.t.}\hspace{10pt}
    &(a)\;\sum_{i=1}^{n_z}\delta_i \leq R_{\mathrm{tot}}\\
    &(b)\;0\leq \delta_i \leq \log_2(\rho_i^{\max}), \hspace{10pt} \forall i=1,\dots,n_z.
\end{split}
\end{equation}
Under this transformation, constraint~(a) becomes affine (and hence convex). Furthermore, the second-order derivative of the objective function $f(\delta_i)$ in \eqref{eq:convexified_problem} is given by
\begin{equation}
    \frac{\partial^2 f(\delta_i)}{\partial \delta_i^2}
    =\ln(2)\bigg[\frac{
    2^{\delta_i}}{(1+2^{\delta_i})^2}
    -\frac{
    \beta^{c}_{n_z}}{\beta^{c}_{n_z}-1} \frac{\lambda_i^2 2^{\delta_i}}{\bigl(1+\lambda_i 2^{\delta_i}\bigr)^2}\bigg].
\end{equation}
Recalling that $\rho_i=2^{\delta_i}$, a sufficient condition for convexity is
\begin{equation}
\left|
   \frac{(1+\rho_i)}{(1+\lambda_i\rho_i)} \right| \leq \frac{1}{\lambda_i}\sqrt{\frac{\beta^{c}_{n_z}-1}{\beta^{c}_{n_z}}}.
\end{equation}
Since $\rho_i \geq 1$ and $\lambda_i \geq 0$, the absolute value can be dropped. Solving the above inequality with respect to $\rho_i$, we obtain the following upper bound on constraint~(\ref{eq:go_gib_vq}-b)
\begin{equation}
\rho_i\leq\rho^{\text{(max)}}_{i}=\frac{\sqrt{\frac{\beta^{c}_{n_z}-1}{\beta^{c}_{n_z}}}-\lambda_i}{\lambda_i\left(1-\sqrt{\frac{\beta^{c}_{n_z}-1}{\beta^{c}_{n_z}}}\right)}.
\end{equation}
It is easy to verify that, for all active modes, the stationary point of the objective function in \eqref{eq:snr_allocation_rule} always satisfies this condition. Therefore, the objective is convex in the relevant region, and the corresponding stationary point is the unique global minimizer.

\section{Proof of Lemma \eqref{lem:monotonicity_lemma}\label{sec:monotonicity_proof}}
\begin{proof}
    Fix a component $i$ and consider two multipliers $\eta_1,\eta_2$ such that $0<\eta_1\le \eta_2$. Let $\rho_1\triangleq\rho_i^{*}(\eta_1)$, $\rho_2\triangleq\rho_i^{*}(\eta_2)$. Where, by definition
    \begin{equation}
        \rho_i(\eta)^{*}=\arg\min_{\rho_i}f(\rho_i)+\eta r_i(\rho_i),
    \end{equation}
    this implies that
    \begin{equation}
        \begin{split}
            f_i(\rho_1)+\eta_1r_i(\rho_1)&\leq f_i(\rho_2)+\eta_1 r_i(\rho_2)\\
            f_i(\rho_2)+\eta_2r_i(\rho_2)&\leq f_i(\rho_1)+\eta_2r_i(\rho_1)\\
        \end{split}
    \end{equation}
    Summing the two inequalities and rearranging terms we obtain
    \begin{equation}
    \begin{split}
        (\eta_1-\eta_2)(r_i(\rho_1)-r_i(\rho_2))&\leq0 \hspace{6pt} \iff r_i(\rho_1) \geq r_i(\rho_2),
    \end{split}
    \end{equation}
    where in the last inequality we exploit the fact that $\eta_1-\eta_2 \leq 0$. Recalling that the rate function is $r_i=\log_2(\rho_i)$, and since the logarithm is monotonically increasing, last equality implies that $\rho_1 \geq \rho_2$, thus proving the lemma.  

\end{proof}

\section{Derivations for the Mutual Information Aware Quantization Scheme}\label{sec:mi_maximization_derivations}
Under the Gaussian approximation of the quantization noise, i.e., 
$\boldsymbol{\eta}_q \sim \mathcal{N}(\mathbf{0},\mathbf{\Sigma}_{\boldsymbol{\eta}_q}),
\quad
\mathbf{\Sigma}_{\boldsymbol{\eta}_q}=\mathrm{diag}(D_1,\dots,D_{n_z})$, maximizing the mutual information $I(\mathbf{z}_q,\mathbf{y})$ leads to a linear mapping \cite[p. 124]{globerson2005minimum}
\begin{equation}
\mathbf{z}_q=\mathbf{W}\mathbf{x} + \boldsymbol{\eta}_q.
\end{equation}
Using the entropy decomposition of mutual information, the optimization problem can be expressed as
\begin{equation}
\begin{split}
\min_{\mathbf{W},\{r_i\}}
&\;\frac{1}{2}\log_2\left(\frac{|\mathbf{W}\mathbf{\Sigma}_{\mathbf{x}|\mathbf{y}}\mathbf{W}^T+\mathbf{\Sigma}_{\boldsymbol{\eta}_q}|}{|\mathbf{W}\mathbf{\Sigma}_x\mathbf{W}^T+\mathbf{\Sigma}_{\boldsymbol{\eta}_q}|}\right)\\
&\text{s.t.} \hspace{6pt} \text{(a)}\;\sum_{i=1}^{n_y}r_i\leq R_{\mathrm{tot}}, \hspace{6pt} \text{(b)}\; r_i\geq 0, \hspace{6pt} \forall i.
\end{split}
\end{equation}
Taking derivatives with respect to $\mathbf{W}$ yields a generalized eigenvalue structure involving $\mathbf{\Sigma}_{\mathbf{x}|\mathbf{y}}\mathbf{\Sigma}_x^{-1}$, implying that the optimal transformation aligns with its eigenvectors. By parameterizing $\mathbf{W}=\mathbf{B}\mathbf{V}$, where $\mathbf{V}$ is the matrix collecting the left-eigenvector of CCA, and $\mathbf{B}=\mathrm{diag}(b_1,\dots,b_{n_y})$ the problem decouples across components. Defining
$s_i = p_i b_i^2$, and $\rho_i = \frac{s_i}{D_i}$ the problem becomes
\begin{equation}
\begin{split}
\min_{\{\rho_i\}_{i=1}^{n_y}}&\sum_{i=1}^{n_y}
\frac{1}{2}\log_2\left(\frac{1+\lambda_i\rho_i}{1+\rho_i}\right)\\
&\text{s.t.} \hspace{6pt} \text{(a)}\;\frac{1}{2}\sum_{i=1}^{n_y}\log_2(\rho_i)\leq R_{\mathrm{tot}}, \hspace{6pt} \text{(b)} \; \rho_i\geq 1 \hspace{6pt} \forall i.
\end{split}
\end{equation}
Taking the logarithm of the optimization variables as in \ref{sec:convexity_analysis}, it can be shown that the problem admits a unique global optimum. 

Then, applying the first order KKT condition, we end-up with the following equation 
\begin{equation}
-\frac{1}{1+\rho_i}
+\frac{\lambda_i}{1+\lambda_i\rho_i}
+\frac{\eta}{\rho_i}
=0,
\end{equation}
which leads to a quadratic equation in $\rho_i$. The optimal solution is obtained by selecting the largest feasible root $\rho_i^\star$, which maximizes the mutual information, i.e.
 \begin{equation}
\rho_i=\left[\frac{-C_i(\eta)+\sqrt{C_i^2(\eta)-4\eta\lambda_i}}{2\eta\lambda_i}\right]_{1}^{\infty},
\end{equation}
where $C_i(\eta)
\triangleq
-\frac{
\lambda_i-1+\eta+\eta\lambda_i
}{
2\eta\lambda_i
}$. The Lagrange multiplier $\eta^\star$ is chosen such that
\begin{equation}
\sum_i r_i(\rho_i^\star)=R_{\mathrm{tot}},
\end{equation}
while inactive components correspond to $\rho_i^\star=1$.

\section{Algorithmic Solution of Problem \ref{eq:block_vq_gib_final}}\label{app:dp_partition}

\begin{algorithm}[ht]
\caption{Dynamic-Programming Solution of Problem~\eqref{eq:block_vq_gib_final}}
\label{alg:dp_partition}
\begin{algorithmic}[1]
\REQUIRE Permuted feature statistics
\(\{s_i^\pi,\lambda_i^\pi\}_{i=1}^{n_y}\), number of blocks \(B\),
maximum block size \(N_{\max}\), block rate \(R_0\), and permutation \(\pi\).
\ENSURE Optimal assignment variables
\(\{\tau_{ib}^\star\}\) and \(\{\alpha_{fb}^\star\}\).
\vspace{0.5mm}
\STATE \textbf{Step 1: Dynamic-Programming Recursion}
\STATE \(J^\star(k,i)=+\infty\), for \(i=1,\ldots,n_y\), and \(k=1,\ldots,K\).
\FOR{\(k=1,\ldots,B\)}
    \FOR{\(i=1,\ldots,n_y\)}
        \FOR{each \(j\in\mathcal{J}(k,i)\)}
            \STATE Compute 
            \(C(j+1,i,R_0)\) using \eqref{eq:block_cost_dp}.
            \STATE $\Gamma = J^\star(k-1,j)+C(j+1,i,R_0)$.
            \IF{\(\Gamma<J^\star(k,i)\)}
                \STATE Set \(J^\star(k,i)=\Gamma\).
                \STATE Set \(P^\star(k,i)=j\).
            \ENDIF
        \ENDFOR
    \ENDFOR
\ENDFOR
\vspace{0.5mm}
\STATE \textbf{Step 2: Boundary Backtracking}
\STATE Set \(c_B^\star=n_y\).
\FOR{\(b=B,\ldots,2\)}
    \STATE Set \(c_{b-1}^\star=P^\star(b,c_b^\star)\).
\ENDFOR
\vspace{0.5mm}
\STATE \textbf{Step 3: Assignment Recovery}
\STATE Set \(c_0^\star=0\).
\FOR{\(b=1,\ldots,B\)}
    \FOR{\(i=1,\ldots,n_y\)}
        \STATE $\tau_{ib}^\star=
        \mathbf{1}\{c_{b-1}^\star<i\le c_b^\star\}$.
    \ENDFOR
\ENDFOR
\FOR{\(f=1,\ldots,n_y\)}
    \FOR{\(b=1,\ldots,B\)}
        \STATE Set
        $\alpha_{fb}^\star=\tau_{\pi^{-1}(f),b}^\star$
    \ENDFOR
\ENDFOR
\RETURN \(\{\tau_{ib}^\star\}\), \(\{\alpha_{fb}^\star\}\).
\end{algorithmic}
\end{algorithm}

In this section to solve the feature-to-block assignment problem in the block-vector quantization allocation, presented Section \ref{sec:block_vector_vq}.
The proposed procedure follows standard dynamic-programming arguments and is
summarized in Algorithm~\ref{alg:dp_partition}.

Let $\mathcal{B}_{\pi}(j,i)=\{j+1,\ldots,i\}$
denote a contiguous block in the permuted feature domain. We define
$J^{\star}(k,i)$ as the minimum value of the objective of Problem
\eqref{eq:block_vq_gib_final} obtained by partitioning the first $i$ permuted
features into exactly $k$ contiguous blocks. Furthermore, we denote by $P^{\star}(k,i)$ the optimal previous cut position
associated with pair $(k,i)$. That is, if $P^{\star}(k,i)=j$, then the first
$j$ permuted features are partitioned into $k-1$ contiguous blocks, while the
last block is given by $\mathcal{B}_{\pi}(j,i)=\{j+1,\ldots,i\}$.

The cost of assigning
the candidate block $\mathcal{B}_{\pi}(j,i)$ to a single block is defined as
\begin{equation}
\label{eq:block_cost_dp}
C\!\left(j+1,i,R_0\right)=-\frac{1}{2}\sum_{\ell=j+1}^{i}\log_2\left(\frac{1+\rho_{\ell}^{(j,i)}}{1+\lambda_{\ell}^{\pi}\rho_{\ell}^{(j,i)}}\right),
\end{equation}
where $\rho_{\ell}^{(j,i)}=\frac{s_{\ell}^{\pi}}{D_{j:i}(R_0)}$ is the SDR associated with feature $\ell$ when the features
$\{j+1,\ldots,i\}$ are grouped into the same block, with $D_{j:i}(R_0)$ denoting
the distortion induced by assigning rate $R_0$ to the candidate block
$\mathcal{B}_{\pi}(j,i)$.

The optimal solution of the feature-partitioning problem exhibits a recursive
structure. Specifically, for each feature index $i$ and each number of blocks
$k$, the optimal cost of partitioning the first $i$ permuted features into
$k$ blocks is obtained by selecting a previous cut position $j$, partitioning
the first $j$ features into $k-1$ blocks, and assigning the remaining features
$\{j+1,\ldots,i\}$ to the last block $\mathcal{B}_{\pi}(j,i)$. Therefore, the
optimal value and the associated optimal cut position satisfy
\begin{equation}
\label{eq:dp_recursion_gib}
\left\{
\begin{aligned}
J^{\star}(k,i)
&=
\min_{j \in \mathcal{J}(k,i)}
\left\{
J^{\star}(k-1,j)+C\!\left(j+1,i,R_0\right)
\right\},\\
P^{\star}(k,i)
&=
\arg\min_{j \in \mathcal{J}(k,i)}
\left\{
J^{\star}(k-1,j)+C\!\left(j+1,i,R_0\right)
\right\}.
\end{aligned}
\right.
\end{equation}

The limits on the block size induced by constraint $(b)$ in
\eqref{eq:block_vq_gib_final} force the feasible cut positions to lie in the set 
\begin{equation}
\label{eq:feasible_cut_set}
\mathcal{J}(k,i)
\triangleq
\left\{
j \in \mathbb{Z}:
\begin{array}{l}
\max\{k-1,\, i-N_{\max}\} \le j,\\
j \le \min\{(k-1)N_{\max},\, i-1\}
\end{array}
\right\}.
\end{equation}

According to standard dynamic programming, recursion
\eqref{eq:dp_recursion_gib} can be computed by storing, for each number of
blocks $k=1,\ldots,B$ and each prefix length $i=1,\ldots,n_y$, both the
optimal cost $J^{\star}(k,i)$ and the corresponding optimal cut position
$P^{\star}(k,i)$, as detailed at lines 3--12 of Algorithm \ref{alg:dp_partition}.   

Once $J^{\star}(k,i)$ and $P^{\star}(k,i)$ have
been computed for all numbers of blocks and prefix lengths, the optimal feature
partition can be recovered by backtracking through the stored cut positions (see lines 20--26 of Algorithm \ref{alg:dp_partition}).
Starting from the terminal boundary $c_B^{\star}=n_y$, the optimal boundaries
are recursively obtained as
\begin{equation}
\label{eq:optimal_boundaries}
c_{b-1}^{\star}
=
P^{\star}\!\left(b,c_b^{\star}\right),
\qquad
b=B,B-1,\ldots,1.
\end{equation}
Hence, the $b$-th optimal block in the permuted feature domain is $\mathcal{B}_{b}^{\pi,\star}=\left\{c_{b-1}^{\star}+1,c_{b-1}^{\star}+2,\ldots,c_b^{\star}\right\}$. The corresponding binary assignment variable is therefore given by
\begin{equation}
\label{eq:tau_computation}
\tau_{ib}
=
\mathbbm{1}
\left\{
c_{b-1}^{\star} < i \le c_b^{\star}
\right\},
\qquad
i=1,\ldots,n_y,\quad b=1,\ldots,B.
\end{equation}

Finally, the assignment can be mapped back to the original feature domain using
the permutation $\pi$. If $f=\pi(i)$ denotes the original index of the $i$-th
permuted feature, then
\begin{equation}
\label{eq:tau_original_domain}
\alpha_{fb}
=
\tau_{\pi^{-1}(f),b},
\qquad
f=1,\ldots,n_y,\quad b=1,\ldots,B,
\end{equation}
where $\alpha_{fb}=1$ if original feature $f$ is assigned to block $b$, and
$\alpha_{fb}=0$ otherwise.

Since \eqref{eq:dp_recursion_gib} requires searching over at
most $N_{\max}$ feasible previous cut positions for each possible block size and prefix length, the computational
complexity of Algorithm \ref{alg:dp_partition}  is $\mathcal{O}(BN_{\max}n_y)$.
\end{appendices}

\bibliographystyle{IEEEtran}
\bibliography{IEEEabrv,bibliography}

\end{document}